\documentclass[11pt]{article}

\usepackage{tomography}

\begin{document}

\title{Efficient quantum tomography II}

\author{Ryan O'Donnell\thanks{Department of Computer Science, Carnegie Mellon University.  Supported by NSF grant CCF-1618679. \texttt{odonnell@cs.cmu.edu}}
\and
John Wright\thanks{Department of Physics, Massachusetts Institute of Technology. Most of this work performed while affiliated with the Department of Computer Science, Carnegie Mellon University. Partially supported by a Simons Fellowship in Theoretical Computer Science.
\texttt{jswright@mit.edu}}}

\maketitle

\begin{abstract}
  Following~\cite{OW16}, we continue our analysis of: (i) ``Quantum tomography'', i.e., learning a quantum state, i.e., the quantum generalization of learning a discrete probability distribution; (ii)~The distribution of Young diagrams output by the RSK algorithm on random words.  Regarding~(ii), we introduce two powerful new tools:
  \begin{itemize}
  	\item A precise upper bound on the expected length of the longest union of $k$ disjoint increasing subsequences in a random length-$n$ word with letter distribution $\alpha_1 \geq \alpha_2 \geq \cdots \geq \alpha_d$.  Our bound has the correct main term and second-order term, and holds for \emph{all}~$n$, not just in the large-$n$ limit.
     \item A new majorization property of the RSK algorithm that allows one to analyze the Young diagram formed by the \emph{lower} rows $\lambda_k, \lambda_{k+1}, \dots$ of its output.
  \end{itemize}
  These tools allow us to prove several new theorems concerning the distribution of random Young diagrams in the \emph{nonasymptotic} regime, giving concrete error bounds that are optimal, or nearly so, in all parameters. As one example, we give a fundamentally new proof of the celebrated fact that the expected length of the longest increasing sequence in a random length-$n$ permutation is bounded by $2\sqrt{n}$. This is the $k = 1$, $\alpha_i \equiv \frac1d$, $d \to \infty$ special case of a much more general result we prove: the expected length of the $k$th Young diagram row produced by an $\alpha$-random word is $\alpha_k n \pm 2\sqrt{\alpha_kd n}$.

  From our new analyses of random Young diagrams we derive several new results in quantum tomography, including:
  \begin{itemize}
  	\item Learning the eigenvalues of an unknown state to $\eps$-accuracy in Hellinger-squared, chi-squared, or KL distance, using $n = O(d^2/\eps)$ copies.
    \item Learning the top-$k$ eigenvalues of an unknown state to $\eps$-accuracy in Hellinger-squared or chi-squared distance using $n = O(kd/\eps)$ copies or in $\ell_2^2$ distance using $n = O(k/\eps)$ copies.
    \item Learning the optimal rank-$k$ approximation of an unknown state to $\eps$-fidelity (Hellinger-squared distance) using $n = \wt{O}(kd/\eps)$ copies.
   \end{itemize}
   We believe our new techniques will lead to further advances in quantum learning; indeed, they have already subsequently been used for efficient von Neumann entropy estimation.
\end{abstract}

\setcounter{page}{0}
\thispagestyle{empty}
\newpage

\section{Introduction}

The \emph{Robinson--Schensted--Knuth (RSK) algorithm}
is a well-known combinatorial algorithm
with diverse applications throughout mathematics, computer science, and physics.
Given a word~$w$ with~$n$ letters from the alphabet~$[d]$,
it outputs two semistandard Young tableaus $(P,Q) = \rsk{w}$
with common shape given by some Young diagram~$\lambda \in \N^d$ ($\lambda_1 \geq \cdots \geq \lambda_d$).  We write $\lambda = \shRSK{w}$, and mention that $\lambda$ can be defined independently of the RSK algorithm as in Theorem~\ref{thm:grn} below.
In the RSK algorithm, the process generating the first row is sometimes called \emph{patience sorting},
and it is equivalent to the basic dynamic program for computing~$w$'s longest (weakly) increasing subsequence.
\begin{definition}
Given a word $w \in [d]^n$, a \emph{subsequence}
is a sequence of letters $(w_{i_1}, \ldots, w_{i_\ell})$
such that $i_1 < \cdots < i_\ell$.
The \emph{length} of the subsequence is~$\ell$.
We say that the subsequence is \emph{weakly increasing},
or just \emph{increasing}, if $w_{i_1} \leq \cdots \leq w_{i_\ell}$.
We write $\lis{w}$ for the length of the longest weakly increasing subsequence in~$w$.
\end{definition}
\noindent
Hence~$\lambda_1 = \lis{w}$,
a result known as Schensted's Theorem~\cite{Sch61}.
Further rows of~$\lambda$ are characterized by Greene's Theorem
as giving the ``higher order LIS statistics" of~$w$.
\begin{theorem}[\cite{Gre74}] \label{thm:grn}
Suppose $\lambda = \shRSK{w}$.  Then for each $k$,
$\lambda_1 + \cdots +\lambda_k$
is equal to the length of the longest union of~$k$ disjoint (weakly) increasing subsequences in~$w$.
\end{theorem}
\noindent
For background on the RSK algorithm, see e.g.~\cite{Ful97,Rom14} and the references therein.

Many applications involve studying the behavior of the RSK algorithm
when its input is drawn from some random distribution.
A famous case is the uniform distribution over length-$n$ permutations $\bpi \sim S_n$ (in which case $d = n$); here the resulting random Young diagram~$\blambda = \rsk{\bpi}$ is said to have \emph{Plancherel distribution}.
Starting with the work of Ulam~\cite{Ula61},
a line of research has studied the distribution of the longest increasing subsequence of~$\bpi$;
its results are summarized as follows: $\E[\lis{\bpi}] \rightarrow 2\sqrt{n}$ as $n \rightarrow \infty$~\cite{LS77,VK77}
(in fact, $\E[\lis{\bpi}] \leq 2\sqrt{n}$ for all~$n$~\cite{VK85,Pil90}), and the deviations of $\lis{\bpi}$ from this value
can be characterized by the Tracy--Widom distribution from random matrix theory~\cite{BDJ99}.
The RSK algorithm has played a central role in many of these developments,
and these results have been shown to apply not just to the first row $\blambda_1 = \lis{\bpi}$
but also to the entire shape of $\blambda$~\cite{Joh01,BOO00}.
In a different stream of research,
the Plancherel distribution arises naturally in quantum algorithms
which perform Fourier sampling over the symmetric group.
Here, its properties have been used to show that
any quantum algorithm for graph isomorphism
(or, more generally,  the hidden subgroup problem on the symmetric group)
which uses the ``standard approach"
must perform highly entangled measurements across many copies of the coset state~\cite{HRT03,MRS08,HMR+10}.

In this work, we consider a more general setting,
sometimes called the \emph{inhomogeneous random word model},
in which the input to the RSK algorithm
is a random word~$\bw$ whose letters are selected independently from some probability distribution.
\begin{definition}
Given a probability distribution $\alpha = (\alpha_1, \ldots, \alpha_d)$ on alphabet~$[d]$,
an \emph{$n$-letter \mbox{$\alpha$-random} word} $\bw = (\bw_1, \ldots, \bw_n)$,
written as $\bw \sim \alpha^{\otimes n}$,
is a random word in which each letter~$\bw_i$ is independently drawn from~$[d]$ according to~$\alpha$.
The \emph{Schur--Weyl distribution} $\SW{n}{\alpha}$
is the distribution on Young diagrams given by $\blambda = \shRSK{\bw}$.  Although it is not obvious, it is a fact that the distribution $\SW{n}{\alpha}$ does not depend on the ordering of $\alpha$'s components.  Thus unless otherwise stated, we will assume that $\alpha$ is sorted; i.e.,\ $\alpha_1 \geq \cdots \geq \alpha_d$.
\end{definition}
\noindent
(The \emph{homogeneous} random word model is the special case in which $\alpha_i = \frac{1}{d}$,
for each $i \in [d]$.  It is easy to see that in this case,
$\SW{n}{\alpha}$ converges to the Plancherel distribution as $d \rightarrow \infty$.)
Aside from arising naturally in combinatorics and representation theory,
the Schur--Weyl distribution also appears
in a large number of problems in quantum learning and data processing,
as we will see below.

Much of the prior work on the Schur--Weyl distribution has occurred in the \emph{asymptotic} regime,
in which~$d$ and~$\alpha$ are held constant and $n \rightarrow \infty$.
An easy exercise in Chernoff bounds shows that $\lis{\bw}/n\rightarrow \alpha_1$
as $n \rightarrow \infty$.
Generalizing this,
a sequence of works~\cite{TW01,Joh01,ITW01,HX13,Mel12}
have shown that in this regime, $\blambda$
is equal to $(\alpha_1 n,\ldots, \alpha_d n)$ plus some lower-order fluctuations
distributed as the eigenvalues of certain random matrix ensembles.
From these works, we may extract the following ansatz,
coarsely describing the limiting behavior of the rows of~$\blambda$.
\begin{center}
\textbf{Ansatz:} For all $k \in [d]$, $\blambda_k \approx \alpha_k n \pm 2\sqrt{\alpha_k d_k n}$.
\end{center}
Here $d_k$ is the number of times $\alpha_k$ occurs in $(\alpha_1, \ldots, \alpha_d)$.
We survey this literature below in Section~\ref{sec:asymptotic}.

\subsection{A nonasymptotic theory of the Schur--Weyl distribution}\label{sec:nonasymptotic}

In this work, motivated by problems in quantum state learning,
we study the Schur--Weyl distribution in the \emph{nonasymptotic} regime.
Previous efforts in this direction were the works~\cite{HM02,CM06}
and, more extensively, our previous paper~\cite{OW16}.
Our goal is to prove worst-case bounds on the shape of~$\blambda$
which hold for all~$n$, independent of~$d$ and~$\alpha$.
When possible, we would like to translate certain features of the
Schur--Weyl distribution
present in the asymptotic regime --- in particular, the ansatz and its consequences --- down into the nonasymptotic regime.

Clearly, nonasymptotic results cannot depend on the quantity $d_k$,
which can be sensitive to arbitrarily small changes in~$\alpha$
that are undetectable when~$n$ is small.
(Consider especially when~$\alpha$ is uniform versus when~$\alpha$ is uniform but with each entry slightly perturbed.)
Instead, our results are in terms of the quantity
$\min\{1, \alpha_k d\}$, for each $k \in [d]$,
which always upper bounds $\alpha_k d_k$.

Our first result tightly bounds the expected row lengths, in line with the ansatz.
\begin{theorem}\label{thm:row-mean}
For $k \in [d]$, set $\nu_k = \min\{1,\alpha_k d\}$.  Then
\begin{equation*}
\alpha_k n - 2\sqrt{\nu_k n} \leq \E_{\blambda \sim \SW{n}{\alpha}} \blambda_k \leq \alpha_k n + 2\sqrt{\nu_k n},
\end{equation*}
\end{theorem}
This improves on a result from~\cite{OW16}, which showed an upper bound in the $k=1$ case with error $+2\sqrt{2}\sqrt{n}$ for general $\alpha$ and with error $+2\sqrt{n}$ for $\alpha$ the uniform distribution.
Setting $\alpha = (\tfrac{1}{d}, \ldots \tfrac{1}{d})$ and letting $d \rightarrow \infty$,
the $k=1$ case of
Theorem~\ref{thm:row-mean} recovers the above-mentioned celebrated fact that
the length of the longest increasing subsequence of a random permutation of~$n$ is at most~$2\sqrt{n}$
in expectation.
Our result gives only the second proof of this statement
since it was originally proved independently by Vershik and Kerov in 1985~\cite{VK85}
and by Pilpel in 1990~\cite{Pil90}.
Next, we bound the mean-squared error of the estimator $\blambda_k/n$ for~$\alpha_k$.
\begin{theorem}\label{thm:mean-squared-intro}
For $k \in [d]$, set $\nu_k = \min\{1,\alpha_k d\}$.  Then
\begin{equation*}
\E_{\blambda \sim \SW{n}{\alpha}} (\blambda_k - \alpha_k n)^2 \leq O(\nu_k n).
\end{equation*}
\end{theorem}
Again, this is in line with the ansatz.
This theorem can be used to derive tight
bounds (up to constant factors)
on the convergence of the \emph{normalized Young diagram}~$\underline{\blambda} = (\blambda_1/n, \ldots, \blambda_d/n)$ to~$\alpha$
in a variety of distance measures, including Hellinger-squared distance
and the KL and chi-squared divergences.
Now in fact, using related techniques, in~\cite{OW16} we were able 
to prove convergence bounds for some distance measures with stronger constants:
\begin{theorem}[\cite{OW16}]\label{thm:lp-norm}
$\displaystyle \E_{\blambda \sim \SW{n}{\alpha}} \Vert \underline{\blambda} - \alpha \Vert_2^2 \leq \frac{d}{n}$
~and~
$\displaystyle \E_{\blambda \sim \SW{n}{\alpha}} \Vert \underline{\blambda} -\alpha \Vert_1 \leq \frac{d}{\sqrt{n}}$.
\end{theorem}
\noindent
In this work, we extend Theorem~\ref{thm:lp-norm} to other, more challenging distance measures.
\begin{theorem}\label{thm:chi-squared}
Let $d(\cdot,\cdot)$ be any of $\dhellsq{\cdot}{\cdot}$, $\dkl{\cdot}{\cdot}$, or $\dchisq{\cdot}{\cdot}$.  Then
$\displaystyle\E_{\blambda \sim\SW{n}{\alpha}} d(\underline{\blambda},\alpha) \leq \frac{d^2}{n}$.
\end{theorem}
\noindent
Not only are Theorems~\ref{thm:lp-norm} and~\ref{thm:chi-squared} in line with the ansatz, they even have the correct constant factors, as predicted below by Theorem~\ref{thm:homogeneous}
in the asymptotic regime.

Finally, we show similar results for truncated distances,
in which only the top~$k$ entries of~$\blambda$
and the top~$k$ entries of~$\alpha$
are compared with each other.
In~\cite{OW16}, this was carried out for truncated~$\ell_1$ distance.
\begin{theorem}[\cite{OW16}]\label{thm:truncated-ellone}
$\displaystyle \E_{\blambda \sim \SW{n}{\alpha}} \dtvk{k}{\underline{\blambda}}{\alpha} \leq \frac{1.92 k + .5}{\sqrt{n}}$.
\end{theorem}
\noindent
By following the proof of this result, our Theorem~\ref{thm:row-mean} immediately implies the same bound with~$1.5$ in place of~$1.92$.
In addition, we prove similar bounds for truncated~$\ell_2^2$, Hellinger, and chi-squared distances.
\begin{theorem}\label{thm:truncated-elltwo}
$\displaystyle \E_{\blambda \sim \SW{n}{\alpha}} \dltwosqk{k}{\underline{\blambda}}{\alpha} \leq \frac{46k}{n}$.
\end{theorem}
\begin{theorem}  \label{thm:truncated-chi}
Let $d(\cdot,\cdot)$ be either  $\dhellsqk{k}{\cdot}{\cdot}$ or $\dchik{k}{\cdot}{\cdot}$.  Then
$\displaystyle\E_{\blambda \sim\SW{n}{\alpha}} d(\underline{\blambda},\alpha) \leq \frac{46kd}{n}.$
\end{theorem}
\noindent
These results follow the ansatz,
though our techniques are not yet strong enough to achieve optimal constant factors.

\subsection{Techniques}
Our main techniques include a pair of majorization theorems for the RSK algorithm.
Here we refer to the following definition.
\begin{definition}
For $x, y \in \R^d$,
we say that~$x$ \emph{majorizes}~$y$, denotes $x \succ y$
if $x_{[1]} + \cdots + x_{[k]} \geq y_{[1]} + \cdots + y_{[k]}$ for all $k \in [d]$,
with equality for $k = d$.
Here the notation $x_{[i]}$ means the $i$th largest value among the $x_j$'s.
In the case of Young diagrams~$\lambda$ and~$\mu$,
we also use the standard notation $\lambda \unrhd \mu$.
\emph{Weak majorization}, denoted with either $\succ_w$ or $\unrhd_w$,
is the case when the equality constraint may not necessarily hold.
\end{definition}

Several of our results require understanding the behavior
of an individual row~$\blambda_k$, for $k \in [d]$.
However, the RSK algorithm's sequential behavior makes understanding rows after the first quite difficult. 
So instead, we adopt the strategy of proving bounds only for the first row (which can sometimes be done directly),
and then translating them to the $k$th row via the following new theorem.
\begin{theorem}                                     \label{thm:catan-intro}
    Fix an integer $k \geq 1$ and an ordered alphabet~$\calA$. Consider the RSK algorithm applied to some string $x \in \calA^n$.  During the course of the algorithm, some letters of~$x$ get bumped from the $k$th row and inserted into the $(k+1)$th row.  Let $x^{(k)}$ denote the string formed by those letters \emph{in the order they are so bumped}.  On the other hand, let~$\ol{x}$ be the subsequence of~$x$ formed by the letters of~$x^{(k)}$ \emph{in the order they appear in~$x$}.  Then $\shRSK{\ol{x}} \unrhd \shRSK{x^{(k)}}$.
\end{theorem}

Our other key tool is the following result
allowing us to bound how much larger $\blambda_1 + \cdots +\blambda_k$
is than its intended value $\alpha_1 n + \cdots + \alpha_k n$ in expectation.
\begin{theorem} \label{thm:ITW-intro}
	Let $\alpha$ be a sorted probability distribution on~$[d]$ and let $k \in [d]$.  Then for all $n \in \N$,
    \[
    	\E\Bigl[\sum_{i=1}^k \blambda^{(n)}_i\Bigr] - \sum_{i=1}^k \alpha_i n\leq \ITW{k}{\alpha}, \quad\text{where $\ITW{k}{\alpha} = \sum_{i \leq k < j} \frac{\alpha_j}{\alpha_i - \alpha_j}$}.
    \]
    Furthermore, using the notation $\Exc{n}{k}{\alpha}$ for the left-hand side, it holds that $\Exc{n}{k}{\alpha} \nearrow \ITW{k}{\alpha}$ as $n \to \infty$ provided that all $\alpha_i$'s are distinct.
\end{theorem}

\noindent
The fact that $\Exc{n}{1}{\alpha}  \to \ITW{1}{\alpha}$ as $n \to \infty$ when all the $\alpha_i$'s are fixed and distinct was originally proven in by Its, Tracy, and Widom~\cite{ITW01}.  We extend this to the general~$k$ case, and also show that the sequence~$\Exc{n}{1}{\alpha}$ is increasing in~$n$,  so that $\ITW{k}{\alpha}$ is an upper bound for all~$n$. So long as~$\alpha_k$ and~$\alpha_{k+1}$ are sufficiently separated
we have found that $\ITW{k}{\alpha}$ gives a surprisingly accurate bound
on $\E[\blambda_1 + \cdots + \blambda_k] - (\alpha_1 n + \cdots + \alpha_k n)$.
When~$\alpha_k$ and~$\alpha_{k+1}$ are not well-separated,
on the other hand, $\ITW{k}{\alpha}$ can be arbitrarily large.
In this case, we consider a mildly perturbed distribution~$\alpha'$
in which~$\alpha_k'$ and~$\alpha_{k+1}'$ \emph{are} well-separated
and then apply Theorem~\ref{thm:ITW-intro} to $\alpha'$ instead.
Supposing that $\alpha' \succ \alpha$, we may then relate the bounds we get
on~$\SW{n}{\alpha'}$ back to~$\SW{n}{\alpha}$ using Theorem~$1.11$ from~\cite{OW16}.
\begin{theorem}[\cite{OW16}]                                     \label{thm:coupling}
    Let $\alpha$, $\beta \in \R^d$ be sorted probability distributions with $\beta \succ \alpha$. Then for any $n \in \N$ there is a coupling $(\blambda, \bmu)$ of $\SW{n}{\alpha}$ and $\SW{n}{\beta}$ such that $\bmu \unrhd \blambda$ always.
\end{theorem}

\subsection{Quantum state learning}

Our main application of these bounds is to problems in the area of quantum state learning.
Here, one is given~$n$ copies of a mixed state $\rho \in \C^{d \times d}$ and asked to learn some property of~$\rho$.
For example, one might attempt to learn the entire $d \times d$ matrix (\emph{quantum tomography}),
just its spectrum $\alpha = (\alpha_1, \ldots, \alpha_d)$ (\emph{quantum spectrum estimation}),
or some other more specific property such as its von~Neumann entropy, its purity, and so forth.
These problems play key roles in various quantum computing applications,
including current-day verification of experimental quantum devices
and hypothesized future quantum protocols such as entanglement verification.
We allow ourselves arbitrary entangled measurements,
and our goal is to learn while using as few copies~$n$ as possible.

The standard approach to designing
entangled measurements for quantum state learning~\cite{ARS88,KW01}
uses a powerful tool from representation theory called \emph{Schur--Weyl duality}, which states that
\begin{equation*}
(\C^d)^{\otimes n} \cong \bigoplus_{\lambda} \Specht{\lambda} \otimes \Weyl{\lambda}{d}.
\end{equation*}
Here the direct sum ranges over all partitions~$\lambda\vdash n$
of height at most~$d$,
and~$\Specht{\lambda}$ and~$\Weyl{\lambda}{d}$ are the irreps of the symmetric and general linear groups corresponding to~$\lambda$.
Measuring $\rho^{\otimes n}$
according to the projectors $\{\Pi_\lambda\}_\lambda$ corresponding to the $\lambda$-subspaces
is called \emph{weak Schur sampling}
and is the optimal measurement if one is interested only
in learning~$\rho$'s spectrum~$\alpha$ (or some function of~$\alpha$).
The outcome of this measurement is a random~$\blambda$
whose distribution depends only on~$\alpha$; in fact:

\begin{fact}\label{fact:schur-weyl}
When performed on~$\rho^{\otimes n}$,
the measurement outcome~$\blambda$ of weak Schur sampling is distributed exactly as
the Schur--Weyl distribution~$\SW{n}{\alpha}$, where~$\alpha$ is~$\rho$'s spectrum.
\end{fact}
\noindent
(See, for example, the discussion of this in~\cite{OW16}.)
Following weak Schur sampling, $\rho^{\otimes n}$ collapses to the subspace corresponding to~$\blambda$,
and if one wishes to learn about more than just~$\rho$'s spectrum, one must perform a further measurement within this subspace.
An algorithm which does so is said to have performed \emph{strong} Schur sampling.
Note that weak Schur sampling refers to a specific measurement, whereas strong Schur sampling refers to a  class of measurements.

Fact~\ref{fact:schur-weyl},
when paired with our results from Section~\ref{sec:nonasymptotic},
immediately suggests the following algorithm for estimating~$\rho$'s spectrum:
perform weak Schur sampling, receive the outcome~$\blambda$, and output $\underline{\blambda}$.
This is exactly the \emph{empirical Young diagram (EYD) algorithm}
introduced independently by Alicki, Ruckinci, and Sadowski~\cite{ARS88}
and Keyl and Werner~\cite{KW01}.
To date, this is the best known spectrum estimation algorithm,
 and it has recently been proposed for current-day experimental implementation~\cite{BAH+16}.
Our Theorem~\ref{thm:chi-squared} immediately implies the following.
\begin{theorem}\label{thm:spectrum}
The spectrum $\alpha$ can be learned in Hellinger-squared distance, KL divergence, and chi-squared divergence using $n=O(d^2/\eps)$ copies.
\end{theorem}
\noindent
Previously, it was known from the works of Hayashi and Matsumoto~\cite{HM02}
and Christandl and Mitcheson~\cite{CM06} that $n = O(d^2/\eps) \cdot \log(d/\eps^2)$ copies sufficed for KL divergence (and hence for Hellinger-squared).
We note that Theorem~\ref{thm:lp-norm} from~\cite{OW16}
gave learning bounds of $O(d/\eps)$ and $O(d^2/\eps^2)$
for spectrum learning under $\ell_2^2$ and $\ell_1$ distance, respectively.
Combined with the lower bound from~\cite{OW15} showing that the EYD algorithm
requires $n = \Omega(d^2/\eps^2)$ copies for $\ell_1$ learning,
we have given optimal bounds for the EYD algorithm in terms of all five distance metrics.

For the more difficult problem of quantum tomography,
the optimal number of copies needed to learn~$\rho$ in trace distance was recently determined to be $n= \Theta(d^2/\eps^2)$ --- the upper bound from our previous work~\cite{OW16}
and the lower bound from the independent work of Haah et al.~\cite{HHJ+16}.
The optimal complexity of learning~$\rho$ in infidelity --- i.e., outputting an estimate~$\widehat{\brho}$ such that $1-F(\rho, \widehat{\brho}) \leq \eps$ ---
remains open, however.
Essentially the best prior result is by Haah~et~al.~\cite{HHJ+16}, who showed that $n = O(d^2/\eps) \cdot \log(d/\eps)$ copies suffice.
For our results, we find it convenient to work with the very closely related
\emph{quantum Hellinger-squared distance} $\dhellsq{\cdot}{\cdot}$. This is known to be the same as infidelity $1-F(\rho, \widehat{\brho})$ up to a factor of~$2$ (see Section~\ref{sec:prelims} for details), and hence learning in quantum Hellinger-squared distance is equivalent to learning in infidelity up to a small constant.
We show the following theorem.
\begin{theorem}
A state $\rho \in \C^{d \times d}$ can be learned in quantum Hellinger-squared distance with copy complexity
\begin{equation*}
n = O\parens*{\min\left\{\frac{d^2}{\eps} \log\left(\frac{d}{\eps}\right), \ \frac{d^3}{\eps}\right\}}.
\end{equation*}
\end{theorem}
\noindent
The left-hand term in the min gives a new proof of the fidelity bound of Haah~et~al.~\cite{HHJ+16}.
 The right-hand term in the min is new; previously it was known only how to learn~$\rho$ in fidelity using $n=O(h(d)/{\eps})$ copies for some unspecified function~$h(\cdot)$ (see the list of citations in~\cite{HHJ+16}).
Along with our trace distance bound of $n= O(d^2/\eps^2)$ --- which implies a fidelity bound of $O(d^2/\eps^2)$ --- we
now have three incomparable upper bounds on the complexity of fidelity tomography,
none of which match the best known lower bound of $\Omega(d^2/\eps)$ from~\cite{HHJ+16}.
Settling the complexity of fidelity learning remains an important open problem.

To perform full-state tomography, we analyze \emph{Keyl's algorithm}~\cite{Key06}.
After performing weak Schur sampling and receiving a random $\blambda$,
it performs a subsequent measurement in the $\blambda$-subspace
whose measurement outcomes correspond to $d \times d$ unitary matrices.
We denote by $\Keyl{\blambda}{\rho}$ the distribution on unitary matrices observed given $\blambda$ and $\rho$.
The algorithm receives a random \mbox{$\bV \sim \Keyl{\blambda}{\rho}$}
from this measurement and then outputs the density matrix
$\bV \mathrm{diag}(\blambda/n) \bV^\dagger$.
We will only require one fact about this algorithm from~\cite{OW16},
and so we defer the full description of Keyl's measurement and algorithm
to the papers~\cite{Key06, OW16}.

\subsection{Principal component analysis}

Next, we consider natural ``principal component analysis" (PCA)-style versions of the above problems.
Here, rather than learning the whole state or spectrum, the goal
is to learn the ``largest" $k$-dimensional part of the state or spectrum.
These problems arise naturally when the state is ``fundamentally" low rank,
but has been perturbed by a small amount of noise.
For spectrum estimation, this involves learning the first~$k$ $\alpha_i$'s
under the ordering $\alpha_1 \geq \cdots \geq \alpha_d$.
Previous work~\cite{OW16}
used Theorem~\ref{thm:truncated-ellone}
to learn the first~$k$ $\alpha_i$'s
in trace distance using $n = O(k^2/\eps^2)$ copies.
Using our Theorems~\ref{thm:truncated-elltwo} and~\ref{thm:truncated-chi}, we extend this result to other distance measures.
\begin{theorem}
The first~$k$ $\alpha_i$'s can be learned in Hellinger-squared distance or chi-squared divergence using $n=O(kd/\eps)$ copies, and in $\ell_2^2$ distance using $n=O(k/\eps)$ copies.
\end{theorem}

For full-state PCA, the natural variant is to output a rank-$k$ matrix~$\widehat{\rho}$
which is almost as good as the best rank-$k$ approximation to~$\rho$.
For trace distance, the work of~\cite{OW16}
showed that $n=O(kd/\eps^2)$ copies are sufficient to output an estimate
with error at most~$\eps$ more than the error of the best rank-$k$ approximation.
In this work, we show the following fidelity PCA result.
\begin{theorem}\label{thm:pca-intro}
	There is an algorithm that, for any $\rho \in \C^{d \times d}$ and $k \in [d]$,
	outputs a random rank-$k$ (or less) hypothesis $\wh{\brho}$ such that
    \[
    	\E[\Dhellsq{\wh{\brho}}{\rho}] \leq \seqtail{\alpha}{k} + O\parens*{\frac{kdL}{n}}
			+ O\parens*{kL \sqrt{\frac{\seqtail{\alpha}{k}}{n}}},
    \]
    where $L = \min\{k, \ln n\}$ and $\seqtail{\alpha}{k} = \alpha_{k+1} + \cdots + \alpha_d$.
\end{theorem}
\noindent
Let us spend time interpreting this result.
The Hellinger-squared error of the best rank-$k$ approximation to~$\rho$ --- the projection of~$\rho$ to its top-$k$ eigenspace --- is given by $\seqtail{\alpha}{k}$.
When $\rho$ is exactly of rank~$k$, then $\seqtail{\alpha}{k}=0$, and this bound tells us that
\begin{equation*}
n = O\left(\min\left\{\frac{kd}{\eps} \log\left(\frac{d}{\eps}\right), \ \frac{k^2 d}{\eps}\right\}\right)
\end{equation*}
copies are sufficient to learn~$\rho$ up to error~$\eps$.
The left-hand term in the min was shown previously by Haah et al.~\cite{HHJ+16} using different techniques,
whereas the right-hand term is new.
In the case that $\rho$ is \emph{not} rank-$k$,
let us first make the reasonable assumption that $k \leq d/\ln n$.
Then
 \[
    	\E[\Dhellsq{\wh{\brho}}{\rho}] \leq \seqtail{\alpha}{k} + Z_1 + Z_2, \text{ where }
	Z_1 = O\parens*{\frac{kd\ln n}{n}},~Z_2 = O\parens*{\sqrt{\frac{\seqtail{\alpha}{k} k d \ln n}{n}}}.
    \]
Noting that $Z_2$ is the geometric mean of $\seqtail{\alpha}{k}$ and $Z_2$, we get that for any $\delta > 0$,
\[
    	\E[\Dhellsq{\wh{\brho}}{\rho}] \leq (1+\delta) \cdot \seqtail{\alpha}{k} + O_\delta\parens*{\frac{kd\ln n}{n}}.
    \]
Hence, this tells us that $n = O(kd/\eps) \cdot \log(d/\eps)$ copies are sufficient
to learn~$\rho$ to error $(1+\delta) \cdot \seqtail{\alpha}{k} + \eps$ (essentially recovering the exactly rank-$k$ case).
Finally, in the unlikely case of $k > d/\ln n$,  a similar argument shows that
$n = O(kd/\eps) \cdot \log^2(d/\eps)$ copies are sufficient
to learn~$\rho$ to error $(1+\delta) \cdot \seqtail{\alpha}{k} + \eps$.

\subsection{Asymptotics of the Schur--Weyl distribution}\label{sec:asymptotic}

In this section, we survey the known results
on the Schur--Weyl distribution in the asymptotic setting.
Though we are primarily interested in proving convergence results with explicit error bounds
in the nonasymptotic setting,
the asymptotic regime is useful for understanding the high-level features of the Schur--Weyl distribution.
Indeed, the early quantum computing papers~\cite{ARS88,KW01} on this topic operated in this regime.

The earliest theorem in this area is due to Vershik and Kerov~\cite{VK81}, who showed the following:
\begin{theorem}[\cite{VK81}]
Let $\alpha = (\alpha_1, \ldots, \alpha_d)$ be a sorted probability distribution,
and let ${\blambda \sim \SW{n}{\alpha}}$.
Then for all $k \in [d]$, as $n \to \infty$ we have $\blambda_k/n \to \alpha_k$ in probability.
\end{theorem}
\noindent
This theorem has been reproven in a variety of works,
including independently by~\cite{ARS88} and~\cite{KW01}
in the quantum computing literature.

Subsequent work determined the lower-order asymptotics of the Schur--Weyl distribution.
As it turns out, the qualitative features of the distribution depend on whether~$\alpha$ has any repeated values.
The simplest case, when all the $\alpha_i$'s are distinct, was first handled
in the work of Alicki, Rudnicki, and Sadowski~\cite{ARS88}.
\begin{theorem}[\cite{ARS88}]\label{thm:distinct}
Let $\alpha = (\alpha_1, \ldots, \alpha_d)$ be a sorted probability distribution in which every entry is distinct.
Let $\blambda \sim \SW{n}{\alpha}$, and let $(\bg_1, \dots, \bg_d)$ be centered jointly Gaussian random variables with $\Var[\bg_i] = \alpha_i(1-\alpha_i)$ and $\Cov[\bg_i, \bg_j] = -\alpha_i\alpha_j$,
for $i \neq j$.\footnote{This is a degenerate Gaussian distribution, supported on $\sum_i \bg_i = 0$.}
Then as $n \rightarrow \infty$,
\begin{equation*}
\left(\frac{\blambda_1 - \alpha_1 n}{\sqrt{n}}, \ldots, \frac{\blambda_d-\alpha_d n}{\sqrt{n}}\right)
\rightarrow
(\bg_1, \ldots, \bg_d)
\end{equation*}
in distribution.
\end{theorem}
\noindent
In other words, in the case of distinct~$\alpha_i$'s, the Schur--Weyl distribution
acts in the asymptotic regime
like the multinomial distribution with parameter~$\alpha$.
The intuition is that given an $\alpha$-random word~$\bw$,
the count of $1$'s is so much greater than the count of any other letter that the longest increasing subsequence
is not much longer than the all-$1$'s subsequence.
Similarly, the longest pair of disjoint increasing subsequences
is not much longer than the all-$1$'s and all-$2$'s subsequences, and so forth.
This theorem has been reproven many times,
such as by~\cite{HX13, Buf12,Mel12,FMN13}.

On the other hand, when $\alpha$ is \emph{degenerate},
i.e.\ the $\alpha_i$'s are \emph{not} all distinct,
then $\SW{n}{\alpha}$ has a surprisingly non-Gaussian limiting behavior.
The first paper along this line of work was by Baik, Deift, and Johannson~\cite{BDJ99};
it characterized the Plancherel distribution (a special case of the Schur--Weyl distribution)
in terms of the eigenvalues of the \emph{Gaussian unitary ensemble}.
\begin{definition}
The \emph{Gaussian unitary ensemble} $\gue{d}$ is the distribution on $d \times d$ Hermitian matrices $\bX$ in which
(i) $\bX_{i,i} \sim \calN(0,1)$ for each $i \in [d]$,
and (ii) $\bX_{i,j}\sim \calN(0,1)_{\C}$ and $\bX_{j,i} = \overline{\bX_{i,j}}$ for all $i < j \in [d]$.
Here $\calN(0,1)_{\C}$ refers to the \emph{complex} standard Gaussian, distributed as $\calN(0, \frac{1}{2}) + i \calN(0,\frac{1}{2})$.
The \emph{traceless GUE}, denoted $\traceless{d}$,
is the probability distribution on $d \times d$ Hermitian matrices $\bY$ given by
\begin{equation*}
\bY = \bX - \frac{\trace(\bX)}{d} \cdot I,
\end{equation*}
where $\bX \sim \gue{d}$.
\end{definition}
The next fact characterizes the eigenvalues of the traceless GUE in the limit (cf.~\cite{HX13}).
\begin{fact}\label{fact:circle-law}
Given $\bY \sim \traceless{d}$, then as $d \rightarrow \infty$,
\begin{equation*}
\left(\frac{\eig_1(\bY)}{\sqrt{d}}, \ldots, \frac{\eig_d(\bY)}{\sqrt{d}}\right)
\end{equation*}
converges almost surely to the semicircle law with density $\sqrt{4-x^2}/2\pi$, $-2 \leq x \leq 2$,
where $\eig_i(\bY)$ denotes the $i$th largest eigenvalue of~$\bY$.
\end{fact}

The traceless GUE was first used to characterize
the Schur--Weyl distribution in the special case
when $\alpha$ is the uniform distribution,
 the homogeneous random word case.
In this case, Tracy and Widom~\cite{TW01}
showed such a characterization for just the first row~$\blambda_1$,
and Johansson~\cite{Joh01} extended their result to hold for the entire diagram~$\blambda$, as follows
(cf.\ the quantum mechanical proof of this theorem by Kuperberg~\cite{Kup02}).
\begin{theorem}[\cite{Joh01}]\label{thm:homogeneous}
Let $\alpha = (\tfrac{1}{d}, \ldots, \frac{1}{d})$ be the uniform distribution.
Let $\blambda \sim \SW{n}{\alpha}$, and let $\bX \sim \traceless{d}$.
Then as $n \rightarrow \infty$,
\begin{equation*}
\left(\frac{\blambda_1 -  n/d}{\sqrt{n/d}}, \ldots, \frac{\blambda_d-n/d}{\sqrt{n/d}}\right)
\rightarrow
(\eig_1(\bX), \ldots, \eig_d(\bX))
\end{equation*}
in distribution.
\end{theorem}

Using Fact~\ref{fact:circle-law},
we expect that for a typical~$\blambda$,
\begin{equation*}
\blambda_1 \approx \frac{n}{d} + 2\sqrt{n}, \qquad
\blambda_d \approx \frac{n}{d} - 2\sqrt{n}
\end{equation*}
and that the remaining $\blambda_i$'s interpolate between these two values.
(Let us also mention a line of work that has considered the case of uniform~$\alpha$ in the nonasymptotic setting.
Here, rather than fixing~$d$ and letting~$n$ tend towards infinity,
$d$ is allowed to grow to infinity while~$n$ scales as $n=O(d^2)$.
In this case, Biane~\cite{Bia01} has shown a limiting theorem for the shape of~$\blambda$,
and M\'{e}liot~\cite{Mel10a} has characterized the fluctions of~$\blambda$ around its mean
with a certain Gaussian process.
The paper of Ivanov and Olshanski~\cite{IO02}, which proves similar results for the Plancherel distribution,
serves as an excellent introduction to this area.)

In the case of general~$\alpha$ ---
the inhomogeneous random word case ---
it is convenient to group the indices $\{1, \ldots, d\}$ according to the degeneracies of~$\alpha$.

\begin{notation}
Suppose there are~$m$ distinct values among the $\alpha_i$'s, and write $\alpha^{(k)}$ for the $k$th largest distinct value.
We will block the indices as
\begin{equation*}
[1, d] = [1, d^{(1)}] \cup [d^{(1)} + 1, d^{(1)} +  d^{(2)}] \cup \cdots \cup [d-d^{(m)} + 1, d],
\end{equation*}
where every $\alpha_i$ in the $k$th block has the value $\alpha^{(k)}$.
Given a partition $\lambda$ of height~$d$, we will write $\lambda^{(k)}_i$ for the $i$th index in the $k$th block,
i.e.\ $\lambda^{(k)}_i = \lambda_{\seqheadstrict{d}{k}+i}$, where $\seqheadstrict{d}{k} = d^{(1)} + \cdots + d^{(k-1)}$.
(We will only use this notation in this subsection; in particular, for Theorem~\ref{thm:inhomogeneous}.)
\end{notation}
\noindent
In the inhomogeneous case,
Its, Tracy, and Widom~\cite{ITW01} gave a limiting characterization for the first row~$\blambda_1$,
and Houdr\'{e} and Xu~\cite{HX13}, in a work that first appeared in~$2009$,
extended their result to hold for the entire diagram~$\blambda$.
Roughly, their characterization shows that within each block,
$\blambda$ acts GUE-like, as in Theorem~\ref{thm:homogeneous},
but across blocks $\blambda$ acts Gaussian-like, as in Theorem~\ref{thm:distinct}.
We cite here a related theorem of M\'{e}liot~\cite{Mel12}, which cleanly decouples these two limiting effects.
\begin{theorem}[\cite{Mel12}]\label{thm:inhomogeneous}
Let $\alpha = (\alpha_1, \ldots, \alpha_d)$ be a sorted probability distribution.
Let $\blambda \sim \SW{n}{\alpha}$,
let $(\bg_1, \ldots, \bg_m)$ be centered jointly Gaussian random variables with  covariances $\delta_{k \ell} d^{(k)} - d^{(k)} d^{(\ell)} \sqrt{\alpha^{(k)}\alpha^{(\ell)}}$,
and let $\bY^{(k)} \sim \traceless{d^{(k)}}$, for each $k \in [m]$.
Then as $n \rightarrow \infty$,
\begin{equation*}
\left\{\frac{\blambda_i^{(k)}-\alpha^{(k)}n}{\sqrt{\alpha^{(k)}n}}\right\}_{k \in [m], i \in [d^{(k)}]}
\rightarrow
\left\{ \frac{\bg_k}{d^{(k)}} + \eig_i(\bY^{(k)})\right\}_{k \in [m], i \in [d^{(k)}]}
\end{equation*}
in distribution.
\end{theorem}

Note that this theorem recovers Theorem~\ref{thm:distinct}
in the case when all the~$\alpha_i$'s are distinct
and Theorem~\ref{thm:homogeneous}
in the case when~$\alpha$ is the uniform distribution.
More generally, if we define $\blambda[k] = \blambda^{(k)}_1 + \cdots + \blambda^{(k)}_{d^{(k)}}$,
then the random variables $(\blambda[k] - \alpha^{(k)} d^{(k)} n)/\sqrt{\alpha^{(k)} d^{(k)} n}$
converge to Gaussian random variables with covariance $\delta_{k \ell} - \sqrt{\alpha^{(k)}d^{(k)} \alpha^{(\ell)}d^{(\ell)}}$.
Hence, within blocks,~$\blambda$ experiences GUE~fluctuations,
whereas across blocks,~$\blambda$ experiences Gaussian fluctuations.

Theorem~\ref{thm:inhomogeneous} predicts
qualitatively different limiting behaviors
between the cases when two $\alpha_i$'s are exactly equal
and when two $\alpha_i$'s are unequal, even if they are close.
Hence its convergence rate naturally depends on quantities like
\begin{equation*}
\max_{i: \alpha_i \neq \alpha_{i+1}} \left(\frac{1}{\alpha_i - \alpha_{i+1}}\right),
\end{equation*}
and it is therefore not applicable in the nonasymptotic regime.
Nevertheless, we have found it useful when reasoning about the Schur--Weyl distribution;
in particular, by disregarding the Gaussian term in Theorem~\ref{thm:inhomogeneous},
we have our ansatz.

\subsection{Future work}
Bavarian, Mehraban, and Wright
have used the
techniques in this work to study
the accuracy of the empirical entropy estimator
for learning the von Neumann entropy.
In preliminary work~\cite{BMW16}, they have shown the following
result:
\begin{theorem}
The bias of the empirical entropy estimator satisfies
\begin{equation*}
H(\alpha) - \frac{3d^2}{2n} \leq \E_{\blambda \sim\SW{n}{\alpha}} H(\underline{\blambda})\leq H(\alpha).
\end{equation*}
Furthermore, the estimator has mean absolute error
\begin{equation*}
\E_{\blambda \sim \SW{n}{\alpha}}|H(\underline{\blambda})-H(\alpha)|
\leq \frac{3d^2}{2n} + \sqrt{\frac{2+\log(d+e)^2}{n}}.
\end{equation*}
Hence, the empirical entropy
is $\eps$-close to the true von Neumann entropy with high probability
when $n = O(d^2/\eps + \log(d)^2/\eps^2)$.
\end{theorem}
\noindent
This gives an expression similar to both the bias and the mean absolute error of the classical empirical entropy estimator~\cite{WY16}.

\subsection{Organization}

Section~\ref{sec:prelims} contains the preliminaries,
Section~\ref{sec:ITW} contains the proof of Theorem~\ref{thm:ITW-intro},
Section~\ref{sec:schur-weyl} contains our results on the concentration of the Schur-Weyl distribution,
Section~\ref{sec:tomography} contains our tomography results,
and Section~\ref{sec:catan} contains our lower-rows majorization theorem.

\paragraph{Acknowledgments}

We would like to thank Mohammad Bavarian and Saeed Mehraban
for allowing us to mention the entropy estimation result.

\section{Preliminaries} \label{sec:prelims}

Please refer to Section~2 of~\cite{OW15} for many of the definitions and notations used in this paper.  We will also introduce additional notation in this section, and establish some simple results.

\begin{notation}
	Given a sequence $\eta = (\eta_1, \dots, \eta_d)$ we write $\seqhead{\eta}{k} = \eta_1 + \cdots + \eta_k$ and we write $\seqtail{\eta}{k} = \eta_{k+1} + \cdots + \eta_d$.
\end{notation}

The following observation concerning Lipschitz constants of the RSK algorithm is very similar to one made in~\cite[Proposition~2.1]{BL12}:
\begin{proposition}\label{prop:lipschitz}
	Suppose $w,w' \in [d]^n$ differ in exactly one coordinate. Write $\lambda = \shRSK{w}$, $\lambda' = \shRSK{w'}$.  Then:
    \begin{itemize}
    	\item $\abs*{\seqhead{\lambda}{k} - \seqhead{\lambda'}{k}} \leq 1$ for every $k \in [d]$.
        \item $\abs{\lambda_k - \lambda'_k} \leq 2$ for every $k \in [d]$.
    \end{itemize}
\end{proposition}
\begin{proof}
	It suffices to prove the first statement; then, using the $k$ and $k-1$ cases, we get the second statement via the triangle inequality.  Also, by interchanging the roles of $w$ and $w'$, it suffices to prove $\seqhead{\lambda}{k} - \seqhead{\lambda'}{k} \leq 1$.  This follows from Greene's Theorem: $\seqhead{\lambda}{k}$ is the length of the longest disjoint union~$U$ of $k$ increasing subsequences in~$w$. If $w'$ is formed by changing one letter in~$w'$, we can simply delete this letter from~$U$ (if it appears) and get a disjoint union of $k$ increasing subsequences in~$w'$ of length at least $\seqhead{\lambda}{k}-1$.  But Greene's Theorem  implies this is a lower bound on $\seqhead{\lambda'}{k}$.
\end{proof}
\begin{remark}
	The bound of~$2$ in the second statement may be tight; e.g., $\shRSK{232122} = (4,1,1)$, $\shRSK{233122} = (3,3,0)$.
\end{remark}
\begin{proposition}	 \label{prop:alpha-lipschitz}
	Let $\alpha$, $\alpha'$ be probability distributions on~$[d]$ and let $\blambda \sim \SW{n}{\alpha}$, $\blambda' \sim \SW{n}{\alpha'}$. Then:
    \begin{itemize}
    	\item $\abs*{\E[\seqhead{\ul{\blambda}}{k}] - \E[\seqhead{\ul{\blambda}'}{k}]} \leq \dtv{\alpha}{\alpha'}$ for every $k \in [d]$.
        \item $\abs*{\E[\ul{\blambda}_{k}] - \E[\ul{\blambda}'_{k}]} \leq 2\dtv{\alpha}{\alpha'}$ for every $k \in [d]$.
    \end{itemize}
\end{proposition}
\begin{proof}
	Again, it suffices to prove the first statement, as the second one easily follows. Write $\eps = \dtv{\alpha}{\alpha'}$.  Thus there is a coupling $(\ba, \ba')$ such that $\ba \sim \alpha$, $\ba' \sim \alpha'$, and $\Pr[\ba \neq \ba'] = \eps$.  Making~$n$ independent draws from the coupled distribution and calling the resulting words $(\bw, \bw')$, it follows that $\E[\triangle(\bw,\bw')] = \eps n$, where $\triangle$ denotes Hamming distance.  Thus repeated application of Proposition~\ref{prop:lipschitz} yields $\abs*{\E[\seqhead{{\blambda}}{k}] - \E[\seqhead{{\blambda'}}{k}]} \leq \eps n$, where $\blambda = \shRSK{\bw}$, $\blambda' = \shRSK{\bw'}$.  But now $\blambda \sim \SW{n}{\alpha}$, $\blambda' \sim \SW{n}{\alpha'}$, so the result follows after dividing through by~$n$.
\end{proof}

The following lemma, while simple, is crucial for our nonasymptotic estimates:
\begin{lemma}										\label{lem:increasing}
	Let $\alpha$ be a probability distribution on $[d]$ with $\alpha_1 \geq \alpha_2 \geq \cdots \geq \alpha_d$.  Fix $k \in [d]$.  Then
    \[
    	\E_{\blambda \sim \SW{n}{\alpha}}\left[\seqhead{\blambda}{k} - \seqhead{\alpha}{k} n\right]
    \]
    is a nondecreasing function of~$n$.
\end{lemma}
\begin{proof}
	We begin by ``reversing'' the alphabet~$[d]$, so that $1 > 2 > \dots > d$; recall that this does not change the distribution of $\blambda$. Further, we will consider all~$n$ simultaneously by letting $\bLambda$ be drawn from the Schur--Weyl process associated to $\bw \sim \alpha^{\otimes \infty}$.  %Thus we wish to show
%    \[
%    	\E\left[\seqhead{\bLambda^{(n)}}{k} - \seqhead{\alpha}{k} n\right]
%    \]
%    is increasing in~$n$.
    Now
    \[	\E\left[\seqhead{\bLambda^{(n)}}{k}\right] = \sum_{t=1}^n \Pr[t^{\text{th}} \text{ letter of } \bw \text{ creates a box in the first $k$ rows}].
    \]
    If $\bw_t \in [k]$ (i.e., it is among the $k$ largest letters), then it will surely create a box in the first~$k$ rows. Since this occurs with probability $\seqhead{\alpha}{k}$ for each~$t$, we conclude that
    \[	\E\left[\seqhead{\bLambda^{(n)}}{k} - \seqhead{\alpha}{k} n\right] = \sum_{t=1}^n \Pr[\bw_t > k \text{ and it creates a box in the first $k$ rows}].
    \]
    This is evidently a nondecreasing function of~$n$.
\end{proof}

\subsection{Distance measures}

\begin{definition}
Let $\alpha, \beta \in \R^d$ be probability distributions.  Then
the truncated \emph{Hellinger-squared distance} is given by
\begin{equation*}
\dhellsqk{k}{\alpha}{\beta}
= \dhellk{k}{\alpha}{\beta}^2
= \sum_{i=1}^k (\sqrt{\alpha_i} - \sqrt{\beta_i})^2,
\end{equation*}
and the $k=d$ case gives $\dhell{\alpha}{\beta} = \dhellk{d}{\alpha}{\beta}$
and $\dhellsq{\alpha}{\beta} = \dhellsqk{d}{\alpha}{\beta}$.
The truncated \emph{chi-squared divergence} is given by
\begin{equation*}
\dchik{k}{\alpha}{\beta}
= \sum_{i=1}^k \beta_i \left(\frac{\alpha_i}{\beta_i}-1\right)^2,
\end{equation*}
and the $k=d$ case gives $\dchi{\alpha}{\beta} = \dchik{d}{\alpha}{\beta}$.
The truncated \emph{$\ell_2^2$ distance} is given by
\begin{equation*}
\dltwosqk{k}{\alpha}{\beta} = \sum_{i=1}^k (\alpha_i - \beta_i)^2,
\end{equation*}
and the $k=d$ case gives $\dltwosq{\alpha}{\beta} = \dltwosqk{d}{\alpha}{\beta}$.
Finally, the \emph{Kullback-Liebler (KL) divergence} is given by
\begin{equation*}
\dkl{\alpha}{\beta} = \sum_{i=1}^d \alpha_i \ln\left(\frac{\alpha_i}{\beta_i}\right).
\end{equation*}
\end{definition}

\begin{proposition}\label{prop:comparing-distances}
These distance measures are related as
\begin{equation*}
\dhellsq{\alpha}{\beta} \leq \dchi{\alpha}{\beta},
\quad
\dhellsqk{k}{\alpha}{\beta} \leq \dchik{k}{\alpha}{\beta},
\quad
\text{and }\dkl{\alpha}{\beta} \leq \dchi{\alpha}{\beta}.
\end{equation*}
\end{proposition}
\begin{proof}
For the first and second inequalities, the bound follows term-by-term:
\begin{equation*}
(\sqrt{\alpha_i} - \sqrt{\beta_i})^2
= \beta_i \left(\sqrt{\frac{\alpha_i}{\beta_i}} - 1\right)^2
\leq \beta_i \left(\sqrt{\frac{\alpha_i}{\beta_i}} - 1\right)^2\left(\sqrt{\frac{\alpha_i}{\beta_i}} + 1\right)^2
= \beta_i \left(\frac{\alpha_i}{\beta_i} - 1\right)^2.
\end{equation*}
On the other hand, the third inequality is proven considering the whole sum at once:
\begin{equation*}
\sum_{i=1}^d \alpha_i \ln\left(\frac{\alpha_i}{\beta_i}\right)
\leq \sum_{i=1}^d \alpha_i \left(\frac{\alpha_i}{\beta_i} - 1\right)
= \sum_{i=1}^d \frac{\alpha_i^2}{\beta_i} - 1,
\end{equation*}
where the inequality uses $\ln(x) \leq x-1$ for all $x > 0$,
and it can be checked that the right-most quantity is equal to~$\dchisq{\alpha}{\beta}$.
\end{proof}

\begin{definition}
Let $\rho, \sigma$ be density matrices.  The \emph{fidelity} is given by
$
F(\rho, \sigma) = \Vert \sqrt{\rho}\sqrt{\sigma}\Vert_1.
$
Related is the \emph{affinity}, given by $A(\rho, \sigma) = \mathrm{tr}(\sqrt{\rho}\sqrt{\sigma})$.
Finally, the \emph{quantum Hellinger-squared distance} is given by
\begin{equation*}
\dhellsq{\rho}{\sigma} = \dhell{\rho}{\sigma}^2= \mathrm{tr}((\sqrt{\rho}-\sqrt{\sigma})^2) = 2  - 2 A(\rho, \sigma).\footnote{We note that the quantum and classical Hellinger-squared distance are often defined with factors of $\frac{1}{2}$ in front. We have omitted them for simplicity.}
\end{equation*}
By definition, $\dhell{\rho}{\sigma} = \Vert \sqrt{\rho}-\sqrt{\sigma}\Vert_F$, and hence it satisfies the triangle inequality.
\end{definition}

\begin{proposition}
These distance measures are related as $F(\rho, \sigma)^2 \leq A(\rho, \sigma) \leq F(\rho, \sigma)$.
As a result, if $\dhellsq{\rho}{\sigma} = \eps$, then $1- \eps/2 \leq F(\rho, \sigma)$.
\end{proposition}
\begin{proof}
The upper bound $A(\rho, \sigma) \leq F(\rho, \sigma)$ is immediate,
as $\mathrm{tr}(M) \leq\Vert M \Vert_1$ for any matrix~$M$.
As for the lower bound, it follows from Equations~$(28)$ and~$(32)$ from~\cite{ANSV08}.
\end{proof}
As a result,  fidelity and affinity are essentially equivalent in the ``$1-\eps$" regime,
and further it suffices to upper bound the Hellinger-squared distance if we want to lower bound the fidelity.
For other properties of the affinity, see~\cite{LZ04,MM15}.
Though we only ever apply the fidelity to density matrices,
we will sometimes apply the affinity to arbitrary positive semidefinite matrices, as in Theorem~\ref{thm:pca-intro}.

\section{Bounding the excess}\label{sec:ITW}
In this section we will study the quantity
\[
    \Exc{n}{k}{\alpha} = \E_{\blambda \sim \SW{n}{\alpha}}\left[\seqhead{\blambda}{k} - \seqhead{\alpha}{k} n\right],
\]
where $\alpha$ is a sorted probability distribution~$[d]$, and $k \in [d]$.  One way to think about this quantity is as
\[
	\E_{\bw \sim \alpha^{\otimes n}}[\seqhead{\blambda}{k} - \seqhead{\bh}{k}],
\]
where $\blambda = \shRSK{\bw}$ and $\bh = \text{Histogram}(\bw)$, i.e.\ $\bh_i$ is the number of~$i$'s in~$\bw$.  By Greene's Theorem we know that $\blambda \succ \bh$ always; thus $\Exc{n}{k}{\alpha} \geq 0$.  We are therefore concerned with upper bounds, trying to quantify how ``top-heavy'' $\blambda$ is on average (compared to a typical~$\bh$).  

As we will see (and as implicitly shown in work of Its, Tracy, and Widom~\cite{ITW01}), the distribution of $\blambda \sim \SW{n}{\alpha}$ is very close to that of a certain modification of the multinomial distribution that favors top-heavy Young diagrams.
\begin{definition}
	For a sorted probability distribution $\alpha$ with all $\alpha_i$'s distinct, define the  function $\Modpdf{n}{\alpha} : \R^n \to \R$ by
\[
	\Modpdf{n}{\alpha}(h) = 1 +  \sum_{1 \leq i < j \leq d} \frac{\alpha_j}{\alpha_i - \alpha_j}\left(\frac{h_i}{\alpha_in} - \frac{h_j}{\alpha_jn}\right).
\]
\end{definition}
For $\bh \sim \Mult{n}{\alpha}$ we have $\E[\bh_\ell] = \alpha_\ell n$; thus $\E[\Modpdf{n}{\alpha}(\bh)] = 1$.  We may therefore  think of $\Modpdf{n}{\alpha}(h)$ as a \emph{relative density} with respect to the $\Mult{n}{\alpha}$ distribution --- except for the fact that we don't necessarily have $\Modpdf{n}{\alpha}(h) \geq 0$ always.  That will not bother us, though; we will only ever compute expectations relative to this density.
\begin{definition}
	We define the \emph{modified $\alpha$-multinomial (signed) distribution} on size-$n$, $d$-letter histograms~$h$ by $\Modpdf{n}{\alpha}(h)\Multpdf{n}{\alpha}(h)$,
	where $\Multpdf{n}{\alpha}(h)$ is the probability of $h$ under $\Mult{n}{\alpha}$.
	 We use the notation
\[
    \E_{\bh \sim \ModMult{n}{\alpha}}[F(\bh)] = \sum_{h} \Modpdf{n}{\alpha}(h)\Multpdf{n}{\alpha}(h) F(h) = \E_{\bh \sim \Mult{n}{\alpha}}[\Modpdf{n}{\alpha}(\bh)F(\bh)].
\]
\end{definition}
As we will see in the proof of Theorem~\ref{thm:our-ITW} below, for each $\lambda \vdash n$,
\begin{equation} \label{eqn:pseudo-approx}
	\text{``}\Pr_{\blambda\sim \SW{n}{\alpha}}[\blambda = \lambda] \approx \Pr_{\bh \sim \ModMult{n}{\alpha}}[\bh = \lambda]\text{''}.
\end{equation}
\begin{remark}
	The modified $\alpha$-multinomial distribution is only defined when $\alpha_1 > \alpha_2 > \cdots > \alpha_d$.  Note that under this condition, a draw $\bh \sim \Mult{n}{\alpha}$ will have $\bh_1 \geq \bh_2 \geq \cdots \geq \bh_d$ with ``very high'' probability, and thus be a genuine partition $\bh \vdash n$.  (The ``very high'' here is only when $n$ is sufficiently large compared to all of the $\frac{1}{\alpha_k - \alpha_{k+1}}$ values, though.)
\end{remark}
The approximation~\eqref{eqn:pseudo-approx} is consistent with the ansatz. One can see from the $\left(\frac{\lambda_i}{\alpha_i n} - \frac{\lambda_j}{\alpha_j n} \right)$ part of the formula for $\Modpdf{n}{\alpha}(\lambda)$ that it emphasizes $\lambda$'s that are ``top-heavy''. That is, it gives more probability to $\lambda$'s that exceed their multinomial-expectation at low indices and fall short of their multinomial-expectation at high indices.  Furthermore, one can see from the $\frac{\alpha_j}{\alpha_i - \alpha_j}$ part of the formula that this effect becomes more pronounced when two or more $\alpha_\ell$'s tend toward equality.

The utility of~\eqref{eqn:pseudo-approx} is that we can compute certain expectations under the modified multinomial distribution easily and exactly, since it has a simple formula. Of course, we have to concern ourselves with the approximation in~\eqref{eqn:pseudo-approx}; in fact, the error can be quite unpleasant in that it depends on~$d$, and even worse, on the gaps $\alpha_k - \alpha_{k+1}$.  Nonetheless, when it comes to using~\eqref{eqn:pseudo-approx} to estimate~$\Exc{n}{k}{\alpha}$, we will see that the increasing property (Lemma~\ref{lem:increasing}) will let us evade the approximation error.  Toward that end, let us make a definition and some calculations:
\begin{notation}
	For any sorted probability distribution $\alpha$ on $[d]$ and any $k \in [d]$ we write
    \[
    	\ITW{k}{\alpha} = \sum_{i \leq k < j} \frac{\alpha_j}{\alpha_i - \alpha_j}.
    \]
\end{notation}
\begin{remark} \label{rem:ITW-gap}
    We have $\ITW{k}{\alpha} = 0$ if $k = d$, and otherwise $\ITW{k}{\alpha}$ is continuous away from $\alpha_k = \alpha_{k+1}$, where it blows up to~$\infty$.  We also have the following trivial bound, which is useful if the gap $\alpha_k - \alpha_{k+1}$ is large:
	\begin{equation} \label{eqn:triv-ITW-bound}
		\ITW{k}{\alpha} \leq k \seqtail{\alpha}{k}/(\alpha_{k} - \alpha_{k+1}).
    \end{equation}
\end{remark}

Although their proof was a little more elaborate, Its, Tracy, and Widom~\cite{ITW01} proved the following result in the special case of $k = 1$:
\begin{proposition}	\label{prop:excess}
	If $\alpha$ is a sorted probability distribution on~$[d]$ with all $\alpha_i$'s distinct, then
    \[
    	\E_{\blambda \sim \ModMult{n}{\alpha}}\bracks[\big]{\Exc{n}{k}{\alpha}} = \ITW{k}{\alpha}.
    \]
\end{proposition}
\begin{proof}
	By definition we have
    \[
    	\E_{\blambda \sim \ModMult{n}{\alpha}}\bracks[\big]{\blambda_k - \alpha_k n} = \E_{\bh \sim \Mult{n}{\alpha}}[\Modpdf{n}{\alpha}(\bh)(\bh_k - \alpha_k n)].
    \]
    It's convenient to write
    \[
	\Modpdf{n}{\alpha}(h) = 1 +  \sum_{1 \leq i < j \leq d} \frac{\alpha_j}{\alpha_i - \alpha_j}\left(\frac{h_i - \alpha_i n}{\alpha_in} - \frac{h_j - \alpha_j n}{\alpha_jn}\right).
	\]
    Then using the fact that for $\bh \sim \Mult{n}{\alpha}$ we have $\E[\bh_k] = \alpha_k n$, $\Var[\bh_k] = \alpha_k(1-\alpha_k) n$, $\Cov[\bh_k,\bh_\ell]=-\alpha_k\alpha_\ell n$, we easily obtain:
    \[
    	\E_{\blambda \sim \ModMult{n}{\alpha}}\bracks[\big]{\blambda_k - \alpha_k n} = \sum_{j > k} \frac{\alpha_j}{\alpha_k - \alpha_j} - \sum_{i < k} \frac{\alpha_k}{\alpha_i - \alpha_k}.
    \]
    The result follows.
\end{proof}

We now come to the main result of this section:
% For  $\alpha$ a probability distribution on~$[d]$ with $\alpha_1 > \alpha_2 > \cdots > \alpha_d$, Its, Tracy, and Widom~\cite[(3-12)]{ITW01} proved that
% \begin{equation}	\label{eqn:itw1}
% 	\E_{\blambda \sim \SW{n}{\alpha}}[\blambda_1 - \alpha_1 n] = \ITW{1}{\alpha} \pm O(1/\sqrt{n}).
% \end{equation}
% However, this was a purely asymptotic result: the constant hidden in the $O(\cdot)$ depends in an unspecified way on~$\alpha$; in particular, on~$d$ and on the gaps $\alpha_k - \alpha_{k+1}$. An essential observation in our work is that, by leveraging the fact (Lemma~\ref{lem:increasing}) that $\E[\blambda_1 - \alpha_1 n]$ is \emph{increasing} in~$n$, the bound~\eqref{eqn:itw1} actually implies $\ITW{1}{\alpha}$ is an upper bound for \emph{every}~$n$.  In fact, Lemma~\ref{lem:increasing} gives us the same increasing property for every $\E[\seqhead{\blambda}{k} - \seqhead{\alpha}{k} n]$, and by generalizing the asymptotic analysis from~\cite{ITW01} yielding~\eqref{eqn:itw1}, we obtain the following:

\begin{theorem} \label{thm:our-ITW}
	Let $\alpha$ be a sorted probability distribution on~$[d]$ and let $k \in [d]$.  Then for all $n \in \N$,
    \[
    	\E_{\blambda \sim \SW{n}{\alpha}}\left[\seqhead{\blambda}{k} - \seqhead{\alpha}{k} n\right] \leq \ITW{k}{\alpha}.
    \]
    Furthermore, using the notation $\Exc{n}{k}{\alpha}$ for the left-hand side, it holds that $\Exc{n}{k}{\alpha} \nearrow \ITW{k}{\alpha}$ as $n \to \infty$ provided that all $\alpha_i$'s are distinct.
\end{theorem}
\begin{remark}
	We expect that $\Exc{n}{k}{\alpha} \nearrow \ITW{k}{\alpha}$ for all~$\alpha$; however we did not prove this.
\end{remark}
\begin{proof}
	Lemma~\ref{lem:increasing} tells us that $\Exc{n}{k}{\alpha}$ is nondecreasing in~$n$ for all~$\alpha$ and~$k$; thus $\Exc{n}{k}{\alpha} \nearrow L_k(\alpha)$ for some $L_k(\alpha) \in \R \cup \{\infty\}$.  The main claim that will complete the proof is the following (the $k = 1$ case of which was proven in~\cite{ITW01}):
\begin{claim}    \label{claim:ITW}
    For fixed $\alpha$ and~$k$,
    \[%\begin{equation}	\label{eqn:ITW-claim}
    	\Exc{n}{k}{\alpha} = \ITW{k}{\alpha} \pm O(1/\sqrt{n})  \quad \text{provided the $\alpha_i$'s are distinct,}
    \]%\end{equation}
    where the constant hidden in the $O(\cdot)$ may depend on~$\alpha$ in an arbitrary way.
\end{claim}
 This claim establishes $L_k(\alpha) = \ITW{k}{\alpha}$ whenever the $\alpha_i$'s are all distinct.  It remains to observe that when the $\alpha_i$'s are not all distinct, $L_k(\alpha) > \ITW{k}{\alpha}$ is impossible; this is because $\ITW{k}{\alpha}-\Exc{n}{k}{\alpha}$ is a continuous function of~$\alpha$ for each fixed~$n$ (unless $\alpha_k = \alpha_{k+1}$, but in this case $\ITW{k}{\alpha} = \infty$ and there is nothing to prove).

    We now focus on proving Claim~\ref{claim:ITW}, following the analysis in~\cite{ITW01}. We emphasize that the $\alpha_i$'s are now assumed distinct, and the constants hidden in all subsequent $O(\cdot)$ notation may well depend on the $\alpha_i$'s.

    As computed in~\cite[top of p.~255]{ITW01}, for each $\lambda \vdash n$ we have
    %the exact formula
%    \begin{equation}	\label{eqn:ITW-formula}
%    	\Pr_{\blambda \sim \SW{n}{\alpha}}[\blambda = \lambda] = \frac{\Delta(\ell)}{\Delta(\alpha)}\frac{1}{\prod_{1 \leq i \leq j  < d} (\lambda_i + d - j)}\cdot \sum_{\pi \in \symm{d}} \sgn(\pi)\left( \prod_{i=1}^d \alpha_i^{d-\pi^{-1}(i)}\right)M_{\alpha \circ \pi}(\lambda),
%    \end{equation}
%	\[
%	    \Pr_{\blambda \sim \SW{n}{\alpha}}[\blambda = \lambda] = \frac{\Delta(\ell)}{\Delta(\alpha)}\cdot \frac{1}{\prod_{1 \leq i \leq j  < d} (\lambda_i + d - j)}\cdot \left(\prod_{i=1}^d \alpha_i^{d-i}\right)\cdot M_{\alpha}(\lambda) \pm \exp(-\Omega(n)),
%	\]
%    where the following notation has been employed: $\ell = (\ell_1, \dots, \ell_d)$ where $\ell_i = \lambda_i + d - i$; $\Delta(x)$~denotes the Vandermonde determinant $\prod_{i < j} (x_i - x_j)$; and, $M_\alpha(\lambda)$ is the multinomial probability mass function --- i.e., the probability that word $\bw \sim \alpha^{\otimes n}$ has $\lambda_i$ copies of letter~$i$ for $1 \leq i \leq d$.
	\begin{equation} \label{eqn:ITW-formula}
	    \Pr_{\blambda \sim \SW{n}{\alpha}}[\blambda = \lambda] = \left(1 + \frac{1}{\sqrt{n}}\left(\sum_{1 \leq i < j \leq d} \sqrt{\frac{\alpha_j}{\alpha_i}}\frac{\sqrt{\alpha_j} \bxi_i - \sqrt{\alpha_i} \bxi_j}{\alpha_i - \alpha_j}\right) \pm  O(\tfrac{1}{n})\right) \Multpdf{n}{\alpha}(\lambda) \pm e^{-\Omega(n)},
	\end{equation}
    where $\bxi_\ell = (\blambda_\ell - \alpha_\ell n)/\sqrt{\alpha_\ell n}$.  If we now simply substitute in the definition of~$\bxi_\ell$ and do some simple arithmetic, we indeed get the following precise form of~\eqref{eqn:pseudo-approx}:
	\begin{equation}
    	\Pr_{\blambda \sim \SW{n}{\alpha}}[\blambda = \lambda] = \Modpdf{n}{\alpha}(\lambda)\Multpdf{n}{\alpha}(\lambda) \pm  O\parens[\big]{\tfrac{\Multpdf{n}{\alpha}(\lambda)}{n}} \pm e^{-\Omega(n)}. \label{eqn:approx-mod-multinomial}
   \end{equation}
%    This estimate is only proven for all genuine partitions $\lambda \vdash n$.  However it still holds for those $d$-letter, size-$n$ histograms~$\lambda$ that are not partitions (i.e., not sorted).   To verify this, first note that for non-partitions~$\lambda$ the left-hand side of~\eqref{eqn:approx-mod-multinomial} is~$0$. On the other hand, using the fact that $\alpha_1 > \alpha_2 > \cdots > \alpha_d$, a~simple Chernoff/union bound shows that $\Multpdf{n}{\alpha}(\lambda) \leq e^{-\Omega(n)}$ for non-partitions (where certainly the constant in the $\Omega(\cdot)$ depends on all the gaps $\alpha_\ell - \alpha_{\ell+1}$).  Using also that $\Modpdf{n}{\alpha}(\lambda) \leq O(1)$, this completes the verification of~\eqref{eqn:approx-mod-multinomial} for all histograms~$\lambda$.

   Given this, let $F$ be any functional on partitions of~$n$ that is subexponentially bounded in~$n$ (meaning $|F(\lambda)| \leq e^{o(n)}$ for all~$\lambda \vdash n$).  Then
   \begin{align*}
   		\E_{\blambda \sim \SW{n}{\alpha}}[F(\blambda)]
        &=  \E_{\blambda \sim \ModMult{n}{\alpha}}[\bone_{\{\blambda\text{ is sorted}\}} \cdot F(\blambda)] \\ & \qquad \pm O(\tfrac{1}{n}) \cdot  \E_{\blambda \sim \Mult{n}{\alpha}}[\bone_{\{\blambda\text{ is sorted}\}} \cdot |F(\blambda)|] \pm e^{-\Omega(n)},
   \end{align*}
   where in the final error $e^{\Omega(n)}$ we used the subexponential bound on $|F(\lambda)|$ and also absorbed a factor of $e^{O(\sqrt{n})}$, the number of partitions of~$n$. We can further simplify this: Using $\alpha_1 > \alpha_2 > \cdots > \alpha_d$, an easy Chernoff/union bound gives that
   \begin{equation} \label{eqn:unsorted}
   	\Pr_{\blambda \sim \Mult{n}{\alpha}}[\blambda \text{ is not sorted}] \leq e^{\Omega(n)}
   \end{equation}
   (where certainly the constant in the $\Omega(\cdot)$ depends on all the gaps $\alpha_\ell - \alpha_{\ell+1}$).  Thus
   \begin{align}
	   \abs*{\E_{\blambda \sim \ModMult{n}{\alpha}}[\bone_{\{\blambda\text{ is unsorted}\}} \cdot F(\blambda)]} &= \abs*{ 	   \E_{\blambda \sim \Mult{n}{\alpha}}[\bone_{\{\blambda\text{ is unsorted}\}} \cdot \Modpdf{n}{\alpha}(\blambda)F(\blambda)]} \nonumber\\
       &\leq \sqrt{ \E_{\blambda \sim \Mult{n}{\alpha}}[\bone^2_{\{\blambda\text{ is unsorted}\}}]}\sqrt{\Modpdf{n}{\alpha}(\blambda)^2F(\blambda)^2} \leq e^{-\Omega(n)}, \label{eqn:SW-approx1}
   \end{align}
   where we used~\eqref{eqn:unsorted}, the subexponential bound on~$F$, and $\Modpdf{n}{\alpha}(\blambda) \leq O(1)$.  A similar but simpler analysis applies to the first middle term in~\eqref{eqn:SW-approx1}, and we conclude the following attractive form of~\eqref{eqn:pseudo-approx} for subexponentially-bounded~$F$:
   \begin{align*}
   		\E_{\blambda \sim \SW{n}{\alpha}}[F(\blambda)]
        &=  \E_{\blambda \sim \ModMult{n}{\alpha}}[F(\blambda)]   \qquad \pm O(\tfrac{1}{n}) \cdot  \E_{\blambda \sim \Mult{n}{\alpha}}[ |F(\blambda)|] \pm e^{-\Omega(n)}.
   \end{align*}
   Finally, Claim~\eqref{claim:ITW} now follows from Proposition~\ref{prop:excess}, together with the fact that for $\blambda \sim \Mult{n}{\alpha}$,
   \[
   	\E[|\Exc{n}{k}{\blambda}|] \leq \sum_{i=1}^k \sqrt{\E[(\blambda_i - \alpha_i n)^2]} = \sum_{i=1}^k \stddev[\blambda_i] = \sum_{i=1}^k \sqrt{n}\sqrt{\alpha_i(1-\alpha_i)} = O(\sqrt{n}). \qedhere
   \]
\end{proof}

\section{Convergence of the Schur--Weyl distribution}\label{sec:schur-weyl}

In this section, we derive consequences of Theorem~\ref{thm:ITW-intro} and Theorem~\ref{thm:catan-intro}.
To begin, it will help to define two restrictions of a word~$w$.
\begin{notation}
Let $w \in [d]^n$ and let $\lambda = \shRSK{w}$.
We use boldface $\bw$ if $\bw \sim \alpha^{\otimes n}$,
in which case $\blambda \sim \SW{n}{\alpha}$.
\begin{itemize}
\item Write $\geqw{k}$ for the string formed from $w$ by deleting all letters smaller than~$k$, and let $\geqlambda{k} = \shRSK{\geqw{k}}$.
	Then the random variable $\geqblambda{k}$ is distributed
	as $\SW{\geqn}{\geqalpha{k}}$,
	where $\geqn \sim \mathrm{Binomial}(n, \geqsum{k})$ and $\geqalpha{k}= (\alpha_i/\geqsum{k})_{i=k}^d$.
\item Write $\leqw{k}$ for the string formed from $w$ by deleting all letters larger than~$k$,
	and let $\leqlambda{k} = \shRSK{\leqw{k}}$.
	Note that if $(P, Q) = \rsk{w}$, then $\leqlambda{k}$ is the shape of the diagram formed by deleting
	all boxes containing letters larger than~$k$ from~$P$, and hence $\leqlambda{k}_i \leq \lambda_i$ for all~$i$.
	Then the random variable $\leqblambda{k}$ is distributed
	as $\SW{\leqn}{\leqalpha{k}}$,
	where $\leqn \sim \mathrm{Binomial}(n, \leqsum{k})$ and $\leqalpha{k} = (\alpha_i/\leqsum{k})_{i=1}^k$.
\end{itemize}
\end{notation}

We will mainly use the following weaker version of Theorem~\ref{thm:John-conjecture}.
\begin{theorem}\label{thm:last-rows-majorize}
Let $\lambda[k{:}]$ denote the Young diagram formed by rows $k, k+1, k+2, \ldots$ of $\lambda$.
Then $\geqlambda{k} \unrhd_w \lambda[k{:}]$.
\end{theorem}
\begin{proof}
This follows by applying Theorem~\ref{thm:catan-intro} to~$w$ and noting that
the string $\ol{w}$ in that theorem is a substring of $\geqw{k}$.
Hence weak majorization holds trivially.
\end{proof}

\subsection{Bounds on the first and last rows}

\begin{theorem}\label{thm:row-one-upper}
Let $\alpha \in \R^d$ be a sorted probability distribution.  Then
\begin{equation*}
\E_{\blambda \sim \SW{n}{\alpha}} \blambda_1 \leq \alpha_1 n + 2 \sqrt{n}.
\end{equation*}
\end{theorem}

\begin{proof}
Write $g = 1/\sqrt{n}$.
We assume that $\alpha_1 + 2 g \leq 1$, as otherwise the theorem is vacuously true.
Let $\beta\in \R^{d}$ be a sorted probability distribution for which
$\beta_1 = \alpha_1 + g$, $\beta_2 \leq \alpha_2$, and $\beta \succ \alpha$.
Then
\begin{equation*}
\E_{\blambda \sim \SW{n}{\alpha}} \blambda_1
\leq \E_{\bmu \sim \SW{n}{\beta}} \bmu_1
\leq \beta_1 n + \sum_{j > 1} \frac{\beta_j}{\beta_1 - \beta_j}
\leq \beta_1 n + \frac{1}{g}
= \alpha_1 n + n g + \frac{1}{g}
= \alpha_1 n + 2\sqrt{n},
\end{equation*}
where the first step is by Theorem~\ref{thm:coupling} and the second is by Theorem~\ref{thm:ITW-intro}.
\end{proof}

\begin{theorem}\label{thm:row-d-lower}
Let $\alpha \in \R^d$ be a sorted probability distribution.  Then
\begin{equation*}
\E_{\blambda \sim \SW{n}{\alpha}} \blambda_d \geq \alpha_d n - 2 \sqrt{\alpha_d d n}.
\end{equation*}
\end{theorem}
\begin{proof}
Write $g = \sqrt{\alpha_d d/n}$.
We assume that $\alpha_d - 2g \geq 0$, as otherwise the theorem is vacuously true.
Let $\beta \in \R^d$ be a sorted probability distribution
for which $\beta_d = \alpha_d - g$, $\beta_{d-1} \geq \alpha_{d-1}$, and $\beta \succ \alpha$.
Then
\begin{multline*}
\E_{\blambda \sim \SW{n}{\alpha}}[\blambda_1 + \cdots + \blambda_{d-1}]
\leq \E_{\bmu \sim \SW{n}{\beta}}[\bmu_1 + \cdots + \bmu_{d-1}]
\leq \beta_1n + \cdots + \beta_{d-1}n + \sum_{i < d} \frac{\beta_d}{\beta_i - \beta_d}\\
\leq \beta_1n + \cdots + \beta_{d-1}n + \frac{d \alpha_d}{g}
=  \alpha_1n + \cdots + \alpha_{d-1}n + 2\sqrt{\alpha_d d n},
\end{multline*}
where the first inequality is by Theorem~\ref{thm:coupling},
and the second inequality is by Theorem~\ref{thm:ITW-intro}.
As $\blambda_1 + \cdots + \blambda_d = n$, this implies that $\E \blambda_d \geq \alpha_d n - 2\sqrt{\alpha_d d n}$.
\end{proof}

We note that  Theorem~\ref{thm:coupling}
can be replaced by Proposition~\ref{prop:alpha-lipschitz} at just a constant-factor expense.

\subsection{Bounds for all rows}

\begin{theorem}\label{thm:row-k-lower}
Let $\alpha \in \R^d$ be a sorted probability distribution.  Then
\begin{equation*}
\displaystyle \alpha_k n - 2\sqrt{\alpha_k k n} \leq  \E_{\blambda \sim \SW{n}{\alpha}} \blambda_k \leq \alpha_k n + 2 \sqrt{\seqtaileq{\alpha}{k}n}.
\end{equation*}
\end{theorem}

\begin{proof}
For the upper bound, we use Theorem~\ref{thm:last-rows-majorize}:
\begin{equation*}
\E_{\blambda \sim \SW{n}{\alpha}} \blambda_k
\leq \E_{\geqn, \geqblambda{k}} \geqblambda{k}_1
\leq \E_{\geqn} \left[\geqalpha{k}_k\geqn + 2 \sqrt{\geqn}\right]
\leq \geqalpha{k}_k \E \geqn + 2 \sqrt{\E \geqn}
\leq \alpha_k n + 2\sqrt{\geqsum{k}n},
\end{equation*}
where the second step is by Theorem~\ref{thm:row-one-upper} and the third is by Jensen's inequality.

For the lower bound, we use the fact that $\blambda_k \geq \leqblambda{k}_k$:
\begin{equation*}
\E_{\blambda \sim \SW{n}{\alpha}} \blambda_k \geq \E_{\leqn,\leqblambda{k}} \leqblambda{k}_k
\geq \E_{\leqn}\left[\leqalpha{k}_k \leqn - 2 \sqrt{\leqalpha{k}_k k \leqn}\right]
\geq \leqalpha{k}_k \E \leqn - 2\sqrt{\leqalpha{k}_k k \E \leqn}
= \alpha_k n - 2\sqrt{\alpha_k k n},
\end{equation*}
where the second inequality is by Theorem~\ref{thm:row-d-lower},
and the third is by Jensen's inequality.
\end{proof}

Theorem~\ref{thm:row-mean} follows from the fact that $\alpha_k k, \seqtaileq{\alpha}{k} \leq \min\{1, \alpha_k d\}$.

\subsection{Chi-squared spectrum estimation}

\begin{theorem}\label{thm:truncated-squares}
Let $\alpha \in \R^d$ be a sorted probability distribution.  Then for any $k \in [d]$,
\begin{equation*}
\E_{\blambda \sim \SW{n}{\alpha}} \sum_{i=k}^d \blambda_i^2 \leq \sum_{i=k}^d (\alpha_i n)^2 + d \geqsum{k} n.
\end{equation*}
\end{theorem}
\begin{proof}
When $k=1$, this statement is equivalent to Lemma~$3.1$ from~\cite{OW16}.
Hence, we may assume $k > 1$.
By Theorem~\ref{thm:last-rows-majorize},
\begin{multline*}
\E_{\blambda \sim \SW{n}{\alpha}} \sum_{i=k}^d \blambda_i^2
\leq \E_{\geqn, \geqblambda{k}} \sum_{i=1}^{d-k+1} (\geqblambda{k}_i)^2
\leq \E_{\geqn}\left[\sum_{i=k}^{d} (\geqalpha{k}_i \geqn)^2 + (d-k+1) \geqn\right]\\
= \sum_{i=k}^d (\alpha_i n)^2 + \sum_{i=k}^d \alpha_i^2 n \left(\frac{1}{\geqsum{k}}-1\right) + (d-k+1)\geqsum{k}n
\leq \sum_{i=k}^d (\alpha_i n)^2 + d \geqsum{k} n.
\end{multline*}
Here the second inequality used Lemma~$3.1$ from~\cite{OW16},
and the third inequality used $\geqsum{k} \geq \alpha_i$ and $k >1$.
\end{proof}

\begin{theorem}\label{thm:chi-squared-spec}
$\displaystyle
\E_{\blambda \sim \SW{n}{\alpha}} \dchi{\underline{\blambda}}{\alpha} \leq \frac{d^2}{n}.
$
\end{theorem}
\begin{proof}
Write the expectation as
\begin{equation*}
\E \dchi{\underline{\blambda}}{\alpha}
= \frac{1}{n^2} \cdot \E \sum_{i=1}^k \frac{\blambda_i^2}{\alpha_i} - 1.
\end{equation*}
To upper bound the expectation, we can apply Theorem~\ref{thm:truncated-squares}.
\begin{align*}
\E \sum_{i=1}^d \frac{\blambda_i^2}{\alpha_i}
= \sum_{i=1}^d \left(\frac{1}{\alpha_i} - \frac{1}{\alpha_{i-1}}\right) \cdot \E \sum_{j=i}^d \blambda_j^2
&\leq \sum_{i=1}^d \left(\frac{1}{\alpha_i} - \frac{1}{\alpha_{i-1}}\right) \cdot  \sum_{j=i}^d \left((\alpha_jn)^2 + d \alpha_j n\right)\\
& = \sum_{j=1}^d \left((\alpha_j n)^2 + d \alpha_j n\right) \cdot \sum_{i=1}^j \left(\frac{1}{\alpha_i} - \frac{1}{\alpha_{i-1}}\right)\\
&= \sum_{j=1}^d\left((\alpha_j n)^2 + d \alpha_j n\right) \cdot \frac{1}{\alpha_j}
= n^2 + d^2 n.
\end{align*}
Dividing through by $n^2$ and subtracting one completes the proof.
\end{proof}

Combined with~Proposition~\ref{prop:comparing-distances}, Theorem~\ref{thm:chi-squared-spec} implies Theorem~\ref{thm:chi-squared}.

\subsection{Concentration bounds}

In this section,
we show that each row $\blambda_i$
concentrates exponentially around its mean.
We do so using the method of bounded differences.

\begin{proposition}\label{prop:azuma}
Let $\alpha \in \R^d$ be a probability distribution.  Then for any $k \in [d]$,
\begin{equation*}
\Var_{\blambda \sim \SW{n}{\alpha}}[\blambda_k] \leq 16n.
\end{equation*}
\end{proposition}
\begin{proof}
Let $\bw \sim \alpha^{\otimes n}$, and consider the martingale $\bX^{(0)}, \ldots, \bX^{(n)}$ defined as
\begin{equation*}
\bX^{(i)} := \E[\blambda_k \mid \bw_1, \ldots, \bw_i].
\end{equation*}
Note that $\bX^{(0)} = \E \blambda_k$ and $\bX^{(n)} = \shRSK{\bw}_k$.
Furthermore, by Proposition~\ref{prop:lipschitz}, we have that $|\bX^{(i)}-\bX^{(i-1)}| \leq 2$ always, for all $i \in [n]$.
Thus, if we write $\nu_k := \E \blambda_k = \bX^{(0)}$, then by Azuma's inequality
\begin{equation*}
\Pr[ |\blambda_k - \nu_k| \geq t] \leq 2 \exp\left(\frac{-t^2}{8n}\right).
\end{equation*}
We can therefore calculate $\Var[\blambda_k]$ as
\begin{equation*}
\E \left(\blambda_k - \nu_k\right)^2
= \int_{t = 0}^\infty 2t \cdot \Pr[|\blambda_k - \nu_k| \geq t] \cdot \mathrm{d}t
\leq \int_{t = 0}^\infty 4t \exp\left(\frac{-t^2}{8n}\right)\cdot \mathrm{d}t
= \left.-16n \exp\left(\frac{-t^2}{8n}\right)\right|_{t=0}^\infty = 16n.\qedhere
\end{equation*}
\end{proof}

\subsection{Truncated spectrum estimation}

\begin{lemma}\label{lem:reduce-to-lessthan-k}
Let $1 \leq i \leq k \leq d$.  Then
\begin{equation*}
\E_{\blambda \sim \SW{n}{\alpha}} (\blambda_i - \alpha_i n)^2
\leq 2 \E_{\leqn, \leqblambda{k}}(\leqblambda{k}_i - \leqalpha{k}_i \leqn)^2
	+ 44 \geqsum{i}n.
\end{equation*}
\end{lemma}

\begin{proof}
Write $\calG$ for the event that $\blambda_i \geq \alpha_i n$.  Then
\begin{equation}\label{eq:split}
\E_{\blambda} (\blambda_i - \alpha_i n)^2
=
\E_{\blambda}[(\blambda_i - \alpha_i n)^2\cdot \bone[\calG]]
+
\E_{\blambda}[(\blambda_i - \alpha_i n)^2\cdot \bone[\overline{\calG}]].
\end{equation}
When~$\calG$ occurs,
then $(\blambda_i - \alpha_i n)^2 \leq (\geqblambda{i}_1 - \alpha_i n)^2$.  Hence
\begin{multline*}
\E[(\blambda_i - \alpha_i n)^2\cdot \bone[\calG]]
\leq \E(\geqblambda{i}_1 - \alpha_i n)^2
= \E(\geqblambda{i}_1 -\geqalpha{i}_1\geqn + \geqalpha{i}_1\geqn - \alpha_i n)^2\\
\leq 2\E(\geqblambda{i}_1 -\geqalpha{i}_1\geqn)^2 + 2\E(\geqalpha{i}_1\geqn - \alpha_i n)^2
\leq 2\E(\geqblambda{i}_1 -\geqalpha{i}_1\geqn)^2  + \frac{2 n \alpha_i^2}{\geqsum{i}},
\end{multline*}
where the second inequality uses $(x+y)^2 \leq 2x^2 + 2y^2$ for all $x, y \in \R$,
and the third inequality is because $\geqn$ is distributed as $\mathrm{Binomial}(n, \geqsum{i})$.
Given $\geqn$, define $\bnu = \E[\geqblambda{i}_1\mid \geqn]$.
Then
\begin{multline}\label{eq:azuma-it-up}
\E (\geqblambda{i}_1 -\geqalpha{i}_1\geqn)^2
= \E(\geqblambda{i}_1 - \bnu + \bnu - \geqalpha{i}_1\geqn)^2\\
= \E[(\geqblambda{i}_1 - \bnu)^2 + (\bnu - \geqalpha{i}_1\geqn)^2
	+ 2(\geqblambda{i}_1 - \bnu)(\bnu - \geqalpha{i}_1\geqn)]
= \E[(\geqblambda{i}_1 - \bnu)^2 + (\bnu - \geqalpha{i}_1\geqn)^2]\\
\leq 16 \E_{\geqn} \geqn + \E_{\geqn} (\bnu - \geqalpha{i}_1\geqn)^2
= 16 \geqsum{i} n + \E_{\geqn} (\bnu - \geqalpha{i}_1\geqn)^2,
\end{multline}
where the inequality uses Proposition~\ref{prop:azuma}.
Next, we note that because $\geqblambda{i}$
is distributed as $\SW{\geqn}{\geqalpha{i}}$,
$\bnu$ is at least $\geqalpha{i}_1 \geqn$.
Hence, to  upper-bound $ (\bnu - \geqalpha{i}_1\geqn)^2$
we must upper-bound~$\bnu$,
and this can be done by Theorem~\ref{thm:row-one-upper}:
$\bnu = \E[\geqblambda{i}_1\mid \geqn]
\leq \geqalpha{i}_1 \geqn + 2\sqrt{\geqn}$.
Thus, \eqref{eq:azuma-it-up} can be bounded as
\begin{equation*}
16 \geqsum{i} n + \E_{\geqn} (\bnu - \geqalpha{i}_1\geqn)^2
\leq 16 \geqsum{i} n + \E_{\geqn} 4\geqn
= 20 \geqsum{i} n.
\end{equation*}
In summary,
the term in~\eqref{eq:split} corresponding to~$\calG$
is at most $42 \geqsum{i} n$.

As for the other term,
when $\calG$ does not occur,
then $(\blambda_i - \alpha_i n)^2 \leq (\leqblambda{k}_i - \alpha_i n)^2$.
Hence
\begin{multline*}
\E[(\blambda_i - \alpha_i n)^2\cdot \bone[\overline{\calG}]]
\leq
\E (\leqblambda{k}_i - \alpha_i n)^2
=
\E (\leqblambda{k}_i -  \leqalpha{k}_i \leqn + \leqalpha{k}_i \leqn - \alpha_i n)^2\\
\leq
2 \E(\leqblambda{k}_i - \leqalpha{k}_i \leqn)^2 + 2\E(\leqalpha{k}_i \leqn - \alpha_i n)^2
\leq
2 \E (\leqblambda{k}_i - \leqalpha{k}_i \leqn)^2 + \frac{2n\alpha_i^2}{\leqsum{k}},
\end{multline*}
where the second inequality uses $(x+y)^2 \leq 2x^2 + 2y^2$ for all $x, y \in \R$,
and the third inequality is because $\leqn$ is distributed as $\mathrm{Binomial}(n, \leqsum{k})$.
As $\alpha_i^2/\leqsum{k}\leq \geqsum{i}$, the proof is complete.
\end{proof}

\begin{theorem}\label{thm:trunc-elltwo}
$\displaystyle
\E_{\blambda \sim \SW{n}{\alpha}} \dltwosqk{k}{\underline{\blambda}}{\alpha} \leq \frac{46 k}{n}.
$
\end{theorem}
\begin{proof}
Applying Lemma~\ref{lem:reduce-to-lessthan-k} with $j = k$ for all $i \in [k]$,
\begin{multline*}
n^2 \E \dltwosqk{k}{\underline{\blambda}}{\alpha}
=
\E \sum_{i=1}^k (\blambda_i - \alpha_i n)^2
\leq
2\E \sum_{i=1}^k (\leqblambda{k}_i - \leqalpha{k}_i \leqn)^2 + 44 kn\\
=
2\E_{\leqn}\left[\leqn^2 \E_{\leqblambda{k}} \Vert \lequblambda{k} - \leqalpha{k} \Vert_2^2\right]  + 44 kn
\leq
2\E_{\leqn}\left[ \leqn^2 \left(\frac{k}{\leqn}\right)\right] + 44 kn
\leq
46 kn,
\end{multline*}
where the second inequality is by Theorem~$1.1$ of~\cite{OW16}.
The theorem follows by dividing through by~$n^2$.
\end{proof}

\begin{theorem}	\label{thm:trunc-chi}
$\displaystyle
\E_{\blambda \sim \SW{n}{\alpha}} \dchik{k}{\underline{\blambda}}{\alpha} \leq \frac{46 kd}{n}.
$
\end{theorem}
\begin{proof}
Applying Lemma~\ref{lem:reduce-to-lessthan-k} with $j = k$ for all $i \in [k]$,
\begin{multline*}
n^2 \E \dchik{k}{\underline{\blambda}}{\alpha}
=
\E \sum_{i=1}^k \frac{1}{\alpha_i} \left(\blambda_i-\alpha_i n\right)^2
\leq
2\E \sum_{i=1}^k \frac{1}{\alpha_i}(\leqblambda{k}_i - \leqalpha{k}_i \leqn)^2
	+ \sum_{i=1}^k \frac{44 \geqsum{i}n}{\alpha_i}\\
\leq
2\E_{\leqn}\left[\frac{\leqn^2}{\leqsum{k}} \E_{\leqblambda{k}} \dchi{\lequblambda{k}}{\leqalpha{k}}\right]
	+44kdn
\leq
2 \E_{\leqn}\left[\frac{\leqn^2}{\leqsum{k}} \cdot \frac{k^2}{\leqn}\right] + 44kdn
=
2 k^2 n + 44 kdn,
\end{multline*}
where the second inequality is because $\geqsum{i} \leq \alpha_i d$,
and the third inequality is by Theorem~\ref{thm:chi-squared-spec}.
The theorem follows from $k \leq d$
and by dividing through by~$n^2$.
\end{proof}

\subsection{Mean squared error}

\begin{theorem}\label{thm:mean-squared}
$\displaystyle
\E_{\blambda \sim \SW{n}{\alpha}} (\blambda_k - \alpha_k n)^2 \leq 42 \alpha_k k n+ 42\geqsum{k}n.
$
\end{theorem}

\begin{proof}
Following the proof of Lemma~\ref{lem:reduce-to-lessthan-k} for $i=k$, we have that
\begin{equation}\label{eq:reuse-old-lemma}
\E_{\blambda} (\blambda_k - \alpha_i n)^2
\leq
\E[(\leqblambda{k}_k - \alpha_k n)^2 \cdot \bone[\overline{\calG}]]
+
42 \geqsum{k} n,
\end{equation}
where $\calG$ is the event that $\blambda_k \geq \alpha_i n$.
Now we borrow a step from the proof of Lemma~$5.1$ in~\cite{OW16}.
Because it has support size~$k$, $\leqalpha{k}$ can be expressed as a mixture
    \begin{equation}        \label{eqn:mixture}
        \leqalpha{k} = p_1 \cdot \calD_1 + p_2 \cdot \calD_2,
    \end{equation}
    of a certain distribution $\calD_1$ supported on~$[k-1]$ and the uniform distribution $\calD_2$ on~$[k]$.
    It can be checked that $p_2 = \leqalpha{k}_k k$.
    We may therefore think of a draw $\leqblambda{k}$ from $\SW{\leqn}{\leqalpha{k}}$ occurring as follows.  First, $[\leqn]$ is partitioned into two subsets $\bI_1, \bI_2$ by including each $i \in [\leqn]$ into $\bI_j$ independently with probability~$p_j$. Next we draw strings $\bw^{(j)} \sim \calD_j^{\otimes \bI_j}$ independently for $j \in [2]$.  Finally, we let $\leqbw{k} = (\bw^{(1)}, \bw^{(2)}) \in [d]^n$ be the natural composite string and define $\leqblambda{k} = \shRSK{\leqbw{k}}$.  Let us also write $\blambda^{(j)} = \shRSK{\blambda^{(j)}}$ for $j \in [2]$.  We now claim that
\begin{equation}\label{eq:split-into-two}
        \sum_{i=1}^z \leqblambda{k}_i \leq \sum_{i=1}^z \blambda_i^{(1)}  + \sum_{i=1}^z \blambda_i^{(2)}
\end{equation}
    always holds. Indeed, this follows from Greene's Theorem: the left-hand side is $|\bs|$, where $\bs \in [d]^n$ is a maximum-length disjoint union of~$z$ increasing subsequences in~$\bw$; the projection of $\bs^{(j)}$ onto coordinates~$\bI_j$ is a disjoint union of~$z$ increasing subsequences in $\bw^{(j)}$ and hence the right-hand side is at least $|\bs^{(1)}| + |\bs^{(2)}| = |\bs|$.

Applying~\eqref{eq:split-into-two} in the $z=k-1$ case,
and using the facts that (i) $|\leqblambda{k}| = |\blambda^{(1)}| + |\blambda^{(2)}|$,
and (ii) $\lambda^{(1)}$ has height at most~$k-1$,
we see that $\blambda_k^{(2)}\leq \leqblambda{k}_k$.
Hence
\begin{equation*}
\E [(\leqblambda{k}_k - \alpha_k n)^2\cdot \bone[\overline{\calG}]]
\leq
\E (\blambda_k^{(2)}- \alpha_k n)^2
=
\E_{\leqn, \bu, \bmu}(\bmu_k- \alpha_k n)^2,
\end{equation*}
where $\bu \sim \mathrm{Binomial}(\leqn,p_2)$ and $\bmu \sim \SW{\bu}{\frac{1}{k}}$.
Hence
\begin{multline*}
\E_{\leqn, \bu, \bmu}(\bmu_k- \alpha_k n)^2
= \E_{\leqn, \bu, \bmu}(\bmu_k- \tfrac{1}{k}\bu + \tfrac{1}{k}\bu - \alpha_k n)^2\\
\leq 2 \E_{\leqn, \bu, \bmu}(\bmu_k- \tfrac{1}{k}\bu)
	+ 2\E_{\leqn, \bu}(\tfrac{1}{k}\bu - \alpha_k n)^2
\leq 2 \E_{\leqn, \bu, \bmu}(\bmu_k- \tfrac{1}{k}\bu)^2
	+ \frac{2 n\alpha_k}{k},
\end{multline*}
where the first inequality used $(x + y)^2 \leq 2 x^2 + 2y^2$.
Given $\bu$, define $\bnu = \E[\bmu_k \mid \bu]$.
Then
\begin{multline}\label{eq:azuma-again???}
\E_{\leqn, \bu, \bmu}(\bmu_k- \tfrac{1}{k}\bu)^2
= \E_{\leqn, \bu, \bmu}(\bmu_k-\bnu + \bnu - \tfrac{1}{k}\bu)^2\\
= \E_{\leqn, \bu, \bmu}[(\bmu_k- \bnu)^2 + (\bnu - \tfrac{1}{k}\bu)^2
	+ 2 (\bmu_k- \bnu)(\bnu - \tfrac{1}{k}\bu)]
= \E_{\leqn, \bu, \bmu}[(\bmu_k- \bnu)^2 + (\bnu - \tfrac{1}{k}\bu)^2]\\
\leq 16\E_{\leqn, \bu} \bu + \E_{\leqn , \bu} (\bnu - \tfrac{1}{k}\bu)^2
= 16\alpha_k k n + \E_{\leqn , \bu} (\bnu - \tfrac{1}{k}\bu)^2,
\end{multline}
where the inequality uses Proposition~\ref{prop:azuma}.
Next, we note that because $\bmu$
is distributed as $\SW{\bu}{\frac{1}{k}}$,
$\bnu$ is at most $\tfrac{1}{k} \bu$.
Hence, to  upper-bound $ (\bnu - \frac{1}{k} \bu)^2$
we must lower-bound~$\bnu$,
and this can be done by Theorem~\ref{thm:row-d-lower}:
$\bnu = \E[\bmu_k \mid \bu]
\geq \frac{1}{k} \bu - 2\sqrt{\bu}$.
Thus, \eqref{eq:azuma-again???} can be bounded as
\begin{equation*}
16 \alpha_k k n + \E_{\leqn, \bu} (\bnu - \tfrac{1}{k}\bu)^2
\leq 16 \alpha_k k n + \E_{\leqn, \bu} 4\bu
= 20 \alpha_k k n.
\end{equation*}
In summary,
the term in~\eqref{eq:reuse-old-lemma} corresponding to~$\overline{\calG}$
is at most $42 \geqsum{k} n$.
\end{proof}

Using the fact that $\alpha_k k, \geqsum{k}\leq \min\{1, \alpha_k d\}$,
this implies Theorem~\ref{thm:mean-squared-intro}.

\subsection{An alternate bound on $\Exc{n}{k}{\alpha}$}

If the gap $\alpha_k - \alpha_{k+1}$ is very tiny (or zero), $\ITW{k}{\alpha}$ will not be a good bound on $\Exc{n}{k}{\alpha}$.  In~\cite{OW16} we gave the following bound:
\begin{proposition}										\label{prop:OW-Exc}
	(\cite[Lemma~5.1]{OW16}.)      $\displaystyle \E_{\blambda \sim \SW{n}{\alpha}}[\seqhead{\ul{\blambda}}{k} - \seqhead{\alpha}{k}] \leq \frac{2\sqrt{2} k}{\sqrt{n}}$.
\end{proposition}
By summing our Theorem~\ref{thm:row-mean} over all $i \in [k]$,
we can replace the constant~$2\sqrt{2}$ by~$2$.
We now observe that this bound can also be improved so that it tends to~$0$ with  $\seqtail{\alpha}{k}$.

\begin{proposition}										\label{prop:ITW-bound-generally}
    $\displaystyle \E_{\blambda \sim \SW{n}{\alpha}}[\seqhead{\ul{\blambda}}{k} - \seqhead{\alpha}{k}] \leq O\parens*{\frac{k \sqrt{\seqtail{\alpha}{k}}}{ \sqrt{n}}}$.
\end{proposition}
\begin{proof}
By Theorem~\ref{thm:coupling}
we may assume that for some~$m \geq 1$ we have $\alpha_k = \alpha_{k+1} = \alpha_{k+2} = \cdots = \alpha_{k+m-1}$ and $\alpha_{k+m+1} = \alpha_{k+m+2} = \cdots = \alpha_d = 0$.
    \paragraph{Case 1:} $m < k$.  In this case, Theorem~\ref{thm:row-k-lower} tells us that
    $\E[\alpha_\ell - \ul{\blambda}_\ell] \leq 2 \sqrt{\alpha_\ell \ell/n}$.
    If $m = 1$ then we are done.  Otherwise, $\alpha_k m \leq 2\seqtail{\alpha}{k}$, and so
    \begin{equation*}
    \E[\seqtail{\alpha}{k} - \seqtail{\ul{\blambda}}{k}]
    	\leq m \cdot 2\sqrt{2} \sqrt{\alpha_k k /n}
	\leq 2\sqrt{2} k \sqrt{\alpha_k m/n}
	\leq 4 k \sqrt{\seqtail{\alpha}{k}/n},
    \end{equation*}
    which is equivalent to our desired bound.

    \paragraph{Case 2:} $m \geq k$. In this case we follow the proof of Proposition~\ref{prop:OW-Exc} from~\cite{OW16}.  Inspecting that proof, we see that in fact the following stronger statement is obtained:
    \[
    	\E_{\blambda \sim \SW{n}{\alpha}}[\seqhead{\ul{\blambda}}{k} - \seqhead{\alpha}{k}] \leq 2k\sqrt{p_2/n} + 2k\sqrt{p_3/n},
    \]
    where it is easy to check (from~\cite[(25)]{OW16}) that %$p_2$ and $p_3$ satisfy $\frac{p_3}{k+m} = \alpha_{k+m}$ and $\frac{p_2}{k+m-1} + \frac{p_3}{k+m} = \alpha_k$.  Thus
    $p_2 +p_3 =  k\alpha_k + \seqtail{\alpha}{k}$. Now
    \[
    	\sqrt{p_2} + \sqrt{p_3} \leq 2\sqrt{p_2 + p_3} \leq 2\sqrt{m \alpha_k + \seqtail{\alpha}{k}} \leq 2\sqrt{2\seqtail{\alpha}{k}}
    \]
    where the middle inequality is because we're in Case~2.  Combining this with the previous inequality completes the proof.
\end{proof}

\section{Tomography with Hellinger/infidelity error} \label{sec:tomography}

\subsection{Setup and notation} \label{sec:fidelity-notation}
In this section we study the number of samples needed in quantum tomography to achieve small quantum Hellinger error (equivalently, infidelity).  Throughout this section we consider the Keyl algorithm, in which
\[
	\blambda \sim \SW{n}{\alpha}, \quad \bV \sim \Keyl{\blambda}{\rho}, \quad \bLambda = \text{diag}(\blambda/n), \quad \text{and the hypothesis is } \wh{\brho} = \bV \bLambda \bV^\dagger,
\]
Later, we will also consider ``PCA''-style results where the output is required to be of rank~$k$.  In that case, $\bLambda$ will be replaced by $\mtxTrunc{\bLambda}{k} = \diag(\blambda_1/n, \dots, \blambda_k/n, 0, \dots, 0)$.

Since the Hellinger distance is unitarily invariant, we have $\Dhellsq{\wh{\brho}}{\rho} = \Dhellsq{\bLambda}{\bR}$, where $\bR = \bV^\dagger \rho \bV$.  It is also an immediate property of the Keyl distribution that the distribution of~$\bR$ depends only on the spectrum of~$\rho$. Thus we may henceforth assume, without loss of generality, that
\begin{equation} \label{eqn:my-notation}
	\rho = A = \diag(\alpha), \quad \text{so } \bV \sim \Keyl{\blambda}{A}, \quad \bR = \bV^\dagger A \bV,
\end{equation}
and our goal is to bound
\begin{equation}	\label{eqn:tomography-target}
	\E_{\blambda, \bV}[\Dhellsq{\bLambda}{\bR}].
\end{equation}
We introduce one more piece of notation. Every outcome $\bV = V$ is a unitary matrix, for which the matrix $(|V_{ij}|^2)_{ij}$ is doubly-stochastic and hence a convex combination of permutation matrices.  We think of $V$ as inducing a random permutation $\bpi$  on~$[d]$, which we write as
\[
	\bpi \sim V.
\]
This arises in expressions like $\bR_{ii} = (\bV^\dagger A \bV)_{ii}$ and $(\sqrt{\bR})_{ii} = (\bV^\dagger \sqrt{A} \bV)_{ii}$, which, by explicit computation, are
\begin{equation}	\label{eqn:Rii}
	\bR_{ii} = \sum_{j=1}^d \abs*{\bV_{ji}}^2 \alpha_j = \E_{\bpi \sim \bV}[\alpha_{\bpi(i)}], \qquad (\sqrt{\bR})_{ii} = \E_{\bpi \sim \bV}[\sqrt{\alpha_{\bpi(i)}}].
\end{equation}
In addition to~\eqref{eqn:my-notation}, we will henceforth always assume $\bpi \sim \bV$.

\subsection{Preliminary tools}
We will require the following theorem from~\cite{OW16}. It is not explicitly stated therein, but it is derived within its ``Proof of Theorem~1.5'' (between the lines labeled ``(30)'' and ``(by (8) again)'').
\begin{theorem}		\label{thm:R}
Let $\rho$ be a $d$-dimensional density matrix~$\rho$ with sorted spectrum~$\alpha$, let $j \in [d]$, let $\blambda \sim \SW{n}{\alpha}$ and let $\bV \sim \Keyl{\blambda}{\rho}$.  Then for  $\bR = \bV^\dagger \rho \bV$, it holds that
\[
	\E\bracks*{\sum_{i=1}^j \bR_{ii}} \geq 2\sum_{i=1}^j \alpha_i - \E_{\blambda' \sim \SW{n+1}{\alpha}}\bracks*{\sum_{i=1}^j \frac{d-i+\blambda'_i}{n+1}}.
\]
\end{theorem}
We will slightly simplify the bound above:
\begin{corollary}	\label{cor:R}
	In the setting described in Section~\ref{sec:fidelity-notation}, and for any $j \in [d]$,
    \[
    	\E\bracks*{\sum_{i=1}^j (\alpha_i - \bR_{ii})} = \E\bracks*{\sum_{i=1}^j (\alpha_i - \alpha_{\bpi(i)})} \leq\E[\seqhead{\ul{\blambda}}{j} - \seqhead{\alpha}{j}] + \frac{jd}{n}.
    \]
\end{corollary}
\begin{proof}
	Here we simply used $\E[\blambda_i'] \leq \E[\blambda_i] + 1$, $d - i + 1 \leq d$, $\frac{1}{n+1} \leq \frac{1}{n}$, and then did some rearranging.
\end{proof}

When it comes to analyzing the quantum Hellinger error of the algorithm, we will end up needing to bound expressions like
\[
	\E\bracks*{\sum_{i=1}^d \parens*{\sqrt{\alpha_i} - \sqrt{\alpha_{\bpi(i)}}}^2} = \E\bracks*{\dhellsq{\alpha \circ \bpi}{\alpha}}.
\]
Ultimately, all we will use about the Keyl distribution on~$\bV$ (and hence the distributions of $\bR$, $\bpi$) is that Corollary~\ref{cor:R} holds.  This motivates the following lemma:
\begin{lemma}	\label{lem:rearrangement}
	Let $\alpha_1 \geq \alpha_2 \geq \cdots \geq \alpha_d > 0$ and let $\beta_1, \dots, \beta_d$ be a permutation of $\alpha_1, \dots, \alpha_d$.  Then
    \[
    	\sum_{i=1}^d (\sqrt{\alpha_i} - \sqrt{\beta_i})^2 \leq 2\sum_{j = 1}^{d-1}\frac{\alpha_j - \alpha_{j+1}}{\alpha_j} \sum_{i=1}^j(\alpha_i - \beta_i).
    \]
\end{lemma}
\begin{proof}
	Let us write $\text{LHS}$ and $\text{RHS}$ for the left-hand side and right-hand side above.  Also, denoting
    $\displaystyle
    	T_i = \sum_{j = i}^{d-1}2\frac{\alpha_j-\alpha_{j+1}}{\alpha_j},
    $
    we have
    $\displaystyle
    	\text{RHS} = \sum_{i=1}^{d-1}(\alpha_i - \beta_i) T_i.
    $

    If $\beta_1, \dots, \beta_d$ is identical to $\alpha_1, \dots, \alpha_d$ then $\text{LHS} = \text{RHS} = 0$.  Otherwise, suppose that $q$ is the least index such that $\alpha_q \neq \beta_q$.  Let $r, s > q$ be such that $\beta_q = \alpha_r$ and $\beta_s = \alpha_q$; then let $\beta'$ denote the permutation~$\beta$ with its $q$th and $s$th entries swapped, so $\beta'_q = \alpha_q$, $\beta'_s = \alpha_r$.  Writing $\text{LHS}'$ and $\text{RHS}'$ for the new values of $\text{LHS}$, $\text{RHS}$, we will show that
    \begin{equation*}	\label{eqn:maple-conj}
    	\text{LHS} - \text{LHS}' \leq \text{RHS} - \text{RHS}'.
    \end{equation*}
    Repeating this argument until $\beta$ is transformed into~$\alpha$ completes the proof.   We have
    \begin{align*}
    	\text{LHS} - \text{LHS}' &= (\sqrt{\alpha_q} - \sqrt{\alpha_r})^2 + (\sqrt{\alpha_q} - \sqrt{\alpha_s})^2 - (\sqrt{\alpha_s} - \sqrt{\alpha_r})^2 = 2(\sqrt{\alpha_q} - \sqrt{\alpha_r})(\sqrt{\alpha_q} - \sqrt{\alpha_s}), \\
    	\text{RHS} - \text{RHS}' &= (\alpha_q-\alpha_r)T_q + (\alpha_s-\alpha_q)T_s - (\alpha_s-\alpha_r)T_s = (\alpha_q - \alpha_r)(T_q-T_s).
    \end{align*}
    Since $s > q$ we have
    \[
    	T_q - T_s = \sum_{j=q}^{s-1}2\frac{\alpha_j - \alpha_{j+1}}{\alpha_j} \geq \frac{2}{\alpha_q} \sum_{j=q}^{s-1}(\alpha_j - \alpha_{j+1}) = \frac{2}{\alpha_q}(\alpha_q - \alpha_s).
    \]
    Thus it remains to show
    \[
    	(\sqrt{\alpha_q} - \sqrt{\alpha_r})(\sqrt{\alpha_q} - \sqrt{\alpha_s}) \leq \frac{(\alpha_q - \alpha_r)(\alpha_q - \alpha_s)}{\alpha_q} = (\sqrt{\alpha_q} - \alpha_r/\sqrt{\alpha_q})(\sqrt{\alpha_q} - \alpha_s/\sqrt{\alpha_q}).
    \]
    This indeed holds, because $\alpha_q \geq \alpha_r \implies \sqrt{\alpha_r} \geq \alpha_r/\sqrt{\alpha_q}$, and similarly for~$s$.
\end{proof}

It will be useful to have a bound on the sum of the ``multipliers'' appearing in Lemma~\ref{lem:rearrangement}.
\begin{lemma}										\label{lem:log}
	Suppose $\alpha_1 \geq \alpha_2 \geq \cdots \geq \alpha_d > 0$.  Let $L = \min\{d, \ln(\alpha_1/\alpha_d)\}$.  Then
    \[
    	\sum_{i=1}^{d-1} \frac{\alpha_i - \alpha_{i+1}}{\alpha_i} \leq L.
    \]
\end{lemma}
\begin{proof}
	The bound of~$d$ is obvious.  Otherwise, the bound involving $\ln(\alpha_1/\alpha_d)$  is equivalent to
    \[
    	\frac{\alpha_1}{\alpha_d} \geq \exp\parens*{\sum_{i=1}^{d-1} \frac{\alpha_i - \alpha_{i+1}}{\alpha_i}} =  \prod_{i=1}^{d-1} \exp\parens*{1 - \frac{\alpha_{i+1}}{\alpha_i}}.
    \]
    But this follows from $\exp(1-z) \leq 1/z$ for $z \in (0,1]$, and telescoping.
\end{proof}

Finally, we also have a variant of Lemma~\ref{lem:rearrangement} that can help if some of the $\alpha_i$'s are very small:
\begin{corollary}										\label{cor:rearrangement}
	Let $\alpha_1 \geq \alpha_2 \geq \cdots \geq \alpha_k > \zeta \geq \alpha_{k+1} \geq \cdots \geq \alpha_d$, where $k < d$, and let $\beta_1, \dots, \beta_d$ be a permutation of $\alpha_1, \dots, \alpha_d$.  Let $\alpha'_i$ be the same as $\alpha_i$ for $i \leq k$, but let $\alpha'_{k+1} = \zeta$.
    \[
    	\sum_{i=1}^d (\sqrt{\alpha_i} - \sqrt{\beta_i})^2 \leq 4\sum_{j = 1}^{k}\frac{\alpha'_j - \alpha'_{j+1}}{\alpha'_j} \sum_{i=1}^j(\alpha_i - \beta_i) + d\zeta + 8kL\zeta,
    \]
    where $L = \min\{k, \ln(\alpha_1/\zeta)\}$.
\end{corollary}
\begin{proof}
	Extend the notation $\alpha'$ by defining ${\alpha}'_i = \max\{\alpha_i, \zeta\}$, and similarly define $\beta'_i$.  Applying Lemma~\ref{lem:rearrangement} we get
    \[
    	\sum_{i=1}^d\parens*{\sqrt{\alpha'_i} - \sqrt{{\beta}'_i}}^2 \leq 2\sum_{j=1}^{k}\frac{\alpha'_j - \alpha'_{j+1}}{\alpha'_j} \sum_{i=1}^j(\alpha'_i - \beta'_i).
    \]
    On the left we can use
    \[
    	\parens*{\sqrt{\alpha_i} - \sqrt{{\beta}_i}}^2 \leq 2\parens*{\sqrt{\alpha'_i} - \sqrt{{\beta}'_i}}^2 +\zeta,
    \]
    which is easy to verify by case analysis.  On the right we can use
\[
\abs*{\sum_{i=1}^j (\alpha_i' - \beta'_i) - \sum_{i=1}^j (\alpha_i - \beta_i)} \leq 2k\zeta
\]
and then the bound from Lemma~\ref{lem:log} applied to the sequence $(\alpha'_1, \dots, \alpha'_{k+1})$.
\end{proof}

\subsection{Tomography analysis}
We begin with with a partial analysis of the general PCA algorithm.

\begin{theorem}										\label{thm:tomography-technical}
	Let $\rho$ be a $d$-dimensional density matrix~$\rho$ with sorted spectrum~$\alpha$, and let $k \in [d]$. Suppose we perform the Keyl algorithm and produce the rank-$k$ (or less) hypothesis $\wh{\brho} = \bV \mtxTrunc{\bLambda}{k}\bV^\dagger$, as described in Section~\ref{sec:fidelity-notation}.  Then
	\[
    	\E[\Dhellsq{\wh{\brho}}{\rho}] \leq \seqtail{\alpha}{k}  + 2\E\bracks*{\dhellsqk{k}{\alpha \circ \bpi}{\alpha}} + 2\E\bracks*{\seqhead{\ul{\blambda}}{k} - \seqhead{\alpha}{k}} + O\parens*{\frac{kd}{n}}.
    \]
\end{theorem}
\begin{proof}
	As described in Section~\ref{sec:fidelity-notation} --- in particular, at~\eqref{eqn:tomography-target} ---  we need to bound $\E[\Dhellsq{\mtxTrunc{\bLambda}{k}}{\bR}]$.  We have
    \begin{align}
	    \E\bracks*{\Dhellsq{\mtxTrunc{\bLambda}{k}}{\bR}} &= \E\bracks*{\tr(\mtxTrunc{\bLambda}{k}) + \tr(\bR) - 2\tr\parens*{\sqrt{\mtxTrunc{\bLambda}{k}} \sqrt{\bR\vphantom{\mtxTrunc{\bLambda}{k}}}}} \nonumber\\
        &= \E\bracks*{\seqtail{\ul{\blambda}}{k} + 2\seqhead{\ul{\blambda}}{k} - 2\sum_{i=1}^k \sqrt{\ul{\blambda}_i} \cdot(\sqrt{\bR})_{ii}} \nonumber\\
        &= \E\bracks*{\seqtail{\ul{\blambda}}{k}} + \E\bracks*{2\seqhead{\ul{\blambda}}{k} - 2\sum_{i=1}^k \sqrt{\ul{\blambda}_i} \sqrt{\alpha_{\bpi(i)}}}\tag{\text{by }\eqref{eqn:Rii}} \\
        &= \E\bracks*{\seqtail{\ul{\blambda}}{k}} + \E\bracks*{\sum_{i=1}^k  \parens*{\sqrt{\ul{\blambda}_i} - \sqrt{\alpha_{\bpi(i)}}}^2} + \E\bracks*{\sum_{i=1}^k(\ul{\blambda}_i - \alpha_{\bpi(i)})}. \label{eqn:123}
 	\end{align}
    We bound the three expressions in~\eqref{eqn:123} as follows:
    \begin{align*}
    	\shortintertext{
        \[
        	 \E\bracks*{\seqtail{\ul{\blambda}}{k}} \leq  \seqtail{\alpha}{k}
 \tag{\text{since $\ul{\blambda} \succ \alpha$}}
 		\] }
    	\E\bracks*{\sum_{i=1}^k  \parens*{\sqrt{\ul{\blambda}_i} - \sqrt{\alpha_{\bpi(i)}}}^2} &\leq 2\E\bracks*{\sum_{i=1}^k  \parens*{\sqrt{\ul{\blambda}_i} - \sqrt{\alpha_{i}}}^2} + 2\E\bracks*{\sum_{i=1}^k  \parens*{\sqrt{\alpha_i} - \sqrt{\alpha_{\bpi(i)}}}^2} \\
        &= 2\E\bracks*{\dhellsqk{k}{\ul{\blambda}}{\alpha}} + 2\E\bracks*{\dhellsqk{k}{\alpha \circ \bpi}{\alpha}}\\
        &\leq  O\parens*{\frac{kd}{n}} + 2\E\bracks*{\dhellsqk{k}{\alpha \circ \bpi}{\alpha}} \tag{\text{by Theorem~\ref{thm:trunc-chi}}}
    	\shortintertext{
        \[
\E\bracks*{\sum_{i=1}^k(\ul{\blambda}_i - \alpha_{\bpi(i)})} = \E\bracks*{\seqhead{\ul{\blambda}}{k} - \seqhead{\alpha}{k}}  + \E\bracks*{\sum_{i=1}^k(\alpha_i - \alpha_{\bpi(i)})} \leq 2\E\bracks*{\seqhead{\ul{\blambda}}{k} - \seqhead{\alpha}{k}} + \frac{kd}{n} \tag{\text{by Corollary~\ref{cor:R}}} \]}
	\end{align*}
    Combining these bounds completes the proof.
\end{proof}

We can now give our analysis for full tomography:
\begin{theorem}	\label{thm:tomtom}
	Suppose $\rho$ has rank~$r$.  Then the hypothesis of the Keyl algorithm satisfies
    \[
    	\E[\Dhellsq{\wh{\brho}}{\rho}] \leq O\parens*{\frac{rd}{n}} \cdot \min\{r, \ln n\}.
    \]
\end{theorem}
\begin{proof}
	When $\rho$ has rank~$r$ we know that $\blambda$ will always have at most~$r$ nonzero rows; thus $\bLambda = \mtxTrunc{\bLambda}{r}$, and we may use the bound in Theorem~\ref{thm:tomography-technical} with $k = r$.  In this case, the terms $\seqtail{\alpha}{r}$ and  $2\E\bracks*{\seqhead{\ul{\blambda}}{r} - \seqhead{\alpha}{r}}$ vanish, the $O(\frac{rd}{n})$ error is accounted for in the theorem statement,  and we will use the simple bound
    \[
    	2\E\bracks*{\dhellsqk{r}{\alpha \circ \bpi}{\alpha}} \leq 2\E\bracks*{\sum_{i=1}^d (\sqrt{\alpha_i} - \sqrt{\alpha_{\bpi(i)}})^2}.
    \]
    We now apply Corollary~\ref{cor:rearrangement} with $\zeta = \frac{r}{n}$; note that the ``$k$'' in that corollary satisfies $k \leq r = \mathop{\mathrm{rank}} \rho$.  We obtain
    \[
    	2\E\bracks*{\sum_{i=1}^d (\sqrt{\alpha_i} - \sqrt{\alpha_{\bpi(i)}})^2} \leq 8\sum_{j = 1}^{k}\frac{\alpha'_j - \alpha'_{j+1}}{\alpha'_j} \E\bracks*{\sum_{i=1}^j(\alpha_i - \alpha_{\bpi(i)})} + 2d\zeta + 16kL\zeta,
    \]
    where $L \leq \min\{k, \ln n\}$.  The latter two terms here are accounted for in the theorem statement, so it suffices to bound the first.  By Corollary~\ref{cor:R} we have
    \[
    	\sum_{j = 1}^{k}\frac{\alpha'_j - \alpha'_{j+1}}{\alpha'_j} \E\bracks*{\sum_{i=1}^j(\alpha_i - \alpha_{\bpi(i)})} \leq \sum_{j = 1}^{k}\frac{\alpha'_j - \alpha'_{j+1}}{\alpha'_j}\E[\seqhead{\ul{\blambda}}{j} - \seqhead{\alpha}{j}] + \sum_{j = 1}^{k}\frac{\alpha'_j - \alpha'_{j+1}}{\alpha'_j}\frac{jd}{n}.
    \]
    The latter quantity here is at most $\frac{kd}{n} L$ (by Lemma~\ref{lem:log}), which again is accounted for in the theorem statement. As for the former quantity, we use Theorem~\ref{thm:our-ITW} to obtain the bound
    \begin{equation} \label{eqn:final-bound}
    	\sum_{j = 1}^{k}\frac{\alpha'_j - \alpha'_{j+1}}{\alpha'_j}\E[\seqhead{\ul{\blambda}}{j} - \seqhead{\alpha}{j}] \leq \frac1n\sum_{j = 1}^{k}\frac{\alpha'_j - \alpha'_{j+1}}{\alpha'_j}\ITW{j}{\alpha}.
    \end{equation}
    (The quantity $\ITW{j}{\alpha}$ may be~$\infty$, but only if $\alpha_j = \alpha_{j+1}$ and hence $\frac{\alpha'_j - \alpha'_{j+1}}{\alpha'_j} = 0$.  The reader may check that it is sufficient for us to proceed with the convention $0 \cdot \infty = 0$.)
    As $\frac{\alpha'_j - \alpha'_{j+1}}{\alpha'_j} = 1 - \frac{\alpha'_{j+1}}{\alpha'_j} \leq 1 - \frac{\alpha_{j+1}}{\alpha_j}$, we can  replace $\alpha'$ with $\alpha$ in~\eqref{eqn:final-bound}.  Then substituting in the definition of~$\ITW{k}{\alpha}$ yields an upper bound of
    \begin{multline*}
    	\frac1n\sum_{j = 1}^{k}\frac{\alpha_j - \alpha_{j+1}}{\alpha_j}\sum_{i \leq j < \ell}\frac{\alpha_\ell}{\alpha_{i} - \alpha_{\ell}} =  \frac1n \sum_{\substack{i \leq k \\ \ell > i}} \frac{\alpha_\ell}{\alpha_i - \alpha_\ell}\sum_{j = i}^{\min\{k,\ell-1\}} \frac{\alpha_j - \alpha_{j+1}}{\alpha_j} \\
        \leq \frac1n \sum_{\substack{i \leq k \\ \ell > i}} \frac{\alpha_\ell}{\alpha_i - \alpha_\ell}\sum_{j = i}^{\ell - 1} \frac{\alpha_j - \alpha_{j+1}}{\alpha_\ell} = \frac1n\sum_{\substack{i \leq k \\ \ell >i}} 1 \leq \frac{kd}{n},
    \end{multline*}
    which suffices to complete the proof of the theorem.
\end{proof}

Finally, we give our ``PCA''-style bound.
Theorem~\ref{thm:pca-intro}
follows by applying Proposition~\ref{prop:ITW-bound-generally}.
\begin{theorem}
	For any $\rho$ and $k \in [d]$, if we apply the Keyl algorithm but output the rank-$k$ (or less) hypothesis $\wh{\brho} = \bV \mtxTrunc{\bLambda}{k}\bV^\dagger$, we have
    \[
    	\E[\Dhellsq{\wh{\brho}}{\rho}] \leq \seqtail{\alpha}{k} + O\parens*{\frac{kdL}{n}} + O(L) \cdot \E[\seqhead{\ul{\blambda}}{k} - \seqhead{\alpha}{k}],
    \]
    where $L = \min\{k, \ln n\}$.
\end{theorem}
\begin{proof}
	The proof proceeds similarly to that of Theorem~\ref{thm:tomtom} but with $k$ in place of~$r$.  Now the terms $\seqtail{\alpha}{k}$ and $\E[\seqhead{\ul{\blambda}}{k} - \seqhead{\alpha}{k}]$ do not vanish but instead go directly into the error bound.  It remains to bound
\[
	2\E\bracks*{\sum_{i=1}^k (\sqrt{\alpha_i} - \sqrt{\alpha_{\bpi(i)}})^2}.
\]

Fix an outcome for~$\bpi = \pi$ and write the associated permutation $\alpha \circ \pi$ as~$\beta$.  Unfortunately we cannot apply Lemma~\ref{lem:rearrangement} because the subsequence $\beta_1, \dots, \beta_k$ is not necessarily a permutation of the subsequence $\alpha_1, \dots, \alpha_k$.  What we can do instead is the following.  Suppose that some $k'$ of the numbers $\beta_1, \dots, \beta_k$ do not appear within $\alpha_1, \dots, \alpha_k$.  Place these missing numbers, in decreasing order, at the end of the $\alpha$-subsequence, forming the new decreasing subsequence $\ol{\alpha}_1, \dots, \ol{\alpha}_{K}$, where $K = k + k' \leq 2k$ and $\ol{\alpha}_i = \alpha_i$ for $i \leq k$.  Similarly extend the $\beta$-subsequence to $\ol{\beta}_1, \dots, \ol{\beta}_{K}$ by adding in the ``missing'' $\alpha_i$'s; the newly added elements can be placed in any order.  Note that all of the elements added to the $\alpha$-subsequence are less than all the elements added to the $\beta$-subsequence (because $\alpha_k$ must be between them).  Thus we have
\begin{equation} \label{eqn:ub}
	\sum_{i=1}^j (\ol{\alpha}_i - \ol{\beta}_i) \leq \sum_{i=1}^k (\alpha_i - \beta_i) \quad \forall\ j > k.
\end{equation}
We may now apply Corollary~\ref{cor:rearrangement} to $\ol{\alpha}$ and $\ol{\beta}$, with its ``$d$'' set to~$K$ and with $\zeta = \frac{k}{n}$.  We get
\[
   	\sum_{i=1}^K \parens*{\sqrt{\ol{\alpha}_i} - \sqrt{\ol{\beta}_i}}^2 \leq 4\sum_{j = 1}^{K}\frac{\ol{\alpha}'_j - \ol{\alpha}'_{j+1}}{\ol{\alpha}'_j} \sum_{i=1}^j(\ol{\alpha}_i - \ol{\beta}_i) + O\parens*{\frac{k^2L}{n}},
\]
    where $L = \min\{k, \ln n\}$.  We can split the sum over $1 \leq j \leq K$ into $1 \leq j \leq k$ and $k < j \leq K$. The latter sum can be bounded by $4L \sum_{i=1}^k (\alpha_i - \beta_i)$, using Lemma~\ref{lem:log} and~\eqref{eqn:ub}.  The former sum is what we ``would have gotten'' had we been able to directly apply Corollary~\ref{cor:rearrangement}.  That is, we have established
    \[
    	\sum_{i=1}^k \parens*{\sqrt{{\alpha}_i} - \sqrt{{\beta}_i}}^2 \leq \sum_{i=1}^K \parens*{\sqrt{\ol{\alpha}_i} - \sqrt{\ol{\beta}_i}}^2 \leq 4 \sum_{j=1}^k \frac{\alpha_j' - \alpha_{j+1}'}{\alpha'_j} \sum_{i=1}^j(\alpha_i - \beta_j) + O\parens*{\frac{k^2L}{n}} + O(L)\cdot \sum_{i=1}^k(\alpha_i - \beta_i).
    \]
    Now taking this in expectation over~$\bpi$ yields
	\[
    	2\E\bracks*{\sum_{i=1}^k (\sqrt{\alpha_i} - \sqrt{\alpha_{\bpi(i)}})^2} \leq 8\sum_{j = 1}^{k}\frac{\alpha'_j - \alpha'_{j+1}}{\alpha'_j} \E\bracks*{\sum_{i=1}^j(\alpha_i - \alpha_{\bpi(i)})} + O\parens*{\frac{k^2L}{n}} + O(L) \cdot \E\bracks*{\sum_{i=1}^k (\alpha_i - \alpha_{\bpi(i)})}.
    \]
    The first term above is handled exactly as in the proof of Theorem~\ref{thm:tomtom}, and
    Corollary~\ref{cor:R} takes care of the last term.
\end{proof}

\section{The lower-row majorization theorem} \label{sec:catan}

The following theorem refers (in~\eqref{eqn:first-chain-condition}) to some terminology ``curves''.  We will actually not define this terminology until inside the proof of theorem; the reader can nevertheless follow the logical flow without knowing the definition.
\begin{theorem}                             \label{thm:catan-conj}
    Let $b \in \calA^m$ be a string of distinct letters and let $A \subseteq \calA$ be a set of letters in~$b$ deemed ``admissible''.  Let $I_1, \dots, I_c$ be disjoint increasing subsequences (possibly empty) in~$b$ of total length~$L$; we assume also that these subsequences are ``admissible'', meaning they consist only of admissible letters.  Finally, assume the following condition:
    \begin{multline} \label{eqn:first-chain-condition}
        \textnormal{\qquad ``one can draw a set of curves through the $I$'s such that} \\
        \textnormal{all inadmissible letters in~$b$ are southeast of the first curve''.\qquad}
    \end{multline}
    (As mentioned, the terminology used here will be defined later.)

    Let $w \in \calA^{m'}$ be a string of distinct letters with the following property: When the RSK algorithm is applied to~$w$, the letters that get bumped into the second row form the string~$b$ (in the order that they are bumped).

    Then there exists a new set of ``admissible'' letters $A' \supseteq A$ for~$w$, with $|A'| = |A| + \Delta$  (so $\Delta \in \N$), along with disjoint admissible (with respect to~$A'$) increasing subsequences $J_1, \dots, J_c$ in~$w$ of total length $L+\Delta$, such that~\eqref{eqn:first-chain-condition} holds for $w$ and the $J$'s with respect $A'$.
\end{theorem}

We will prove this theorem, as well as the following lemma, later.
\begin{lemma}           \label{lem:catan-init}
    Let $b \in \calA^m$ be a string of distinct letters and let $I_1, \dots, I_c$ be disjoint increasing subsequences in~$b$.  Then there are disjoint increasing subsequences $I'_1, \dots, I'_c$ consisting of the same letters as $I_1, \dots, I_c$, just grouped differently, such that it is possible to draw a ``set of curves'' through $I'_1, \dots, I'_c$.
\end{lemma}

Let us now see what Theorem~\ref{thm:catan-conj} and Lemma~\ref{lem:catan-init} imply:
\begin{theorem}                                     \label{thm:John-conjecture}
    Fix an integer $k \geq 1$. Consider the RSK algorithm applied to some string $x \in \calA^n$.  During the course of the algorithm, some letters of~$x$ get bumped from the $k$th row and inserted into the $(k+1)$th row.  Let $x^{(k)}$ denote the string formed by those letters \emph{in the order they are so bumped}.  On the other hand, let~$\ol{x}$ be the subsequence of~$x$ formed by the letters of~$x^{(k)}$ \emph{in the order they appear in~$x$}.  Then $\shRSK{\ol{x}} \unrhd \shRSK{x^{(k)}}$.
\end{theorem}
\begin{proof}
    We may assume all the letters in $x$ are distinct, by the usual trick of ``standardization''; this does not affect the operation of RSK on~$x$ or~$\ol{x}$.   When RSK is applied to $x$, let us write more generally $x^{(j)}$ ($1 \leq j \leq k$) for the sequence of letters bumped from the~$j$th row and inserted into the $(j+1)$th row, in the order they are bumped.  We also write $x^{(0)} = x$.

    We will show $\shRSK{\ol{x}} \unrhd \shRSK{x^{(k)}}$ using Greene's Theorem; it suffices to show that if $I^{(k)}_1, \dots, I^{(k)}_c$ are any disjoint increasing subsequences in $x^{(k)}$ of total length~$L$, there are some~$c$ disjoint increasing subsequences $\ol{I}_1, \dots, \ol{I}_c$ of total length at least~$L$ in~$\ol{x}$.  We will find these subsequences by applying Theorem~\ref{thm:catan-conj} $k$ times in succession, with $(b,w)$ equal to $(x^{(j)}, x^{(j-1)})$ for $j = k, {k-1}, {k-2}, \dots, 1$.

    In the first application, with $b = x^{(k)}$ and $w = x^{(k-1)}$,  we will declare \emph{all} letters appearing in~$b$ to be ``admissible''.  In particular this means that $I^{(k)}_1, \dots, I^{(k)}_c$ are automatically admissible. After reorganizing these subsequences using Lemma~\ref{lem:catan-init} (if necessary), we may draw \emph{some} ``set of curves'' through them.  Condition~\eqref{eqn:first-chain-condition} is then vacuously true, as there are no inadmissible letters.    Theorem~\ref{thm:catan-conj} thus gives us some $\Delta_k$ newly admissible letters, as well as admissible disjoint increasing subsequences $I^{(k-1)}_1, \dots, I^{(k-1)}_c$ in $x^{(k-1)}$ of total length $L+\Delta_k$, such that condition~\eqref{eqn:first-chain-condition} still holds.

    We now continue applying Theorem~\ref{thm:catan-conj}, $k-1$ more times, until we end up with admissible disjoint increasing subsequences $I^{(0)}_1, \dots, I^{(0)}_c$ in $x^{(0)} = x$ of total length $L+\ol{\Delta}$, where $\ol{\Delta} = \Delta_1 + \cdots + \Delta_k$ is the number of newly admissible letters in~$x$, beyond those letters originally appearing in~$x^{(k)}$.  Finally, we delete all of these newly admissible letters wherever they appear in $I^{(0)}_1, \dots, I^{(0)}_c$, forming $\ol{I}_1, \dots, \ol{I}_c$; these are then disjoint increasing subsequences of total length at least~$L$.  But they also consist only of letters that were originally admissible, i.e.\ in~$x^{(k)}$; hence $\ol{I}_1, \dots, \ol{I}_c$ are subsequences of $\ol{x}$ and we are done.
\end{proof}

We now come to the proof of Theorem~\ref{thm:catan-conj} (which includes the definition of ``curves'').
\begin{proof}[Proof of Theorem~\ref{thm:catan-conj}]
    Our proof of the theorem uses Viennot's geometric interpretation of the RSK process~\cite{Vie81}; see e.g.,~\cite[Chapter 3.6]{Sag01},~\cite[Chapter~2]{Wer94}) for descriptions.  We may assume that $\calA = [D]$ for some~$D \in \N$.  The word~$w = (w_1, \dots, w_n)$ is then identified with its ``graph''; i.e., the set of points $(i, w_{i})$ in the $2$-dimensional plane. We will call these the ``white points''.  (Since $w$ has distinct letters, no two white points are at the same height.)  Viennot's construction then produces a set of ``shadow lines'' through the points; we will call them ``jump lines'', following the terminology in Wernisch~\cite{Wer94}.  The points at the northeast corners are called the ``skeleton'' of~$w$; we will also call these the ``black points''. They are the graph of the string~$b$ (if we use~$w$'s indices when indexing~$b$).

    Note that an increasing subsequence in~$b$ (respectively,~$w$) corresponds to an increasing sequence of black (respectively, white) points; i.e., a sequence in which each successive point is to the northeast of the previous one.  We will call such a sequence a \emph{chain}. In Theorem~\ref{thm:catan-conj} we are initially given~$c$ disjoint sequences/chains $I_1, \dots, I_c$ in~$b$, of total length~$L$.  To aid in the geometrical description, we will imagine that $L$~``beads'' are placed on all of these chain points.\footnote{As a technical point, the theorem allows some of these chains to be empty.  We will henceforth discard all such ``empty'' chains, possibly decreasing~$c$.  In the end we will also allow ourselves to produce fewer than~$c$ chains $J_i$ in~$w$.  But this is not a problem, as we can always artificially introduce more empty chains.}

    We are also initially given a set~$A$ of admissible ``letters'' in~$b$; in the geometric picture, these correspond to ``admissible black points''.  We are initially promised that all the beads are at admissible black points.  The end of the theorem statement discusses admissible letters~$A'$ in~$w$; these will correspond to ``admissible white points''.  Since $A' \supseteq A$, every white point that is directly west of an admissible black point will be an admissible white point.  But in addition, the theorem allows us to designate some~$\Delta$ additional white points as ``newly admissible''.

    Our final goal involves finding~$c$ chains of admissible white points, of total length $L+\Delta$.  The outline for how we will do this is as follows: First we will add some number $\Delta$ of beads to the initial chains, and at the same time designate some~$\Delta$ white points as newly admissible.  Then we will ``slide'' each bead some amount west and north along its jump line, according to certain rules.  At the end of the slidings, all the beads will be on admissible white points, and we will show that they can be organized into~$c$ chains.  Finally, we must show that the postcondition~\eqref{eqn:first-chain-condition} concerning ``curves'' is still satisfied, given that it was satisfied initially.

    Let's now explain what is meant by ``curves'' and ``sets of curves''.  A curve will refer to an infinite southwest-to-northeast curve which is the graph of a strictly increasing continuous function that diverges to $\pm \infty$ in the limits to $\pm \infty$.  It will typically pass through the beads of a chain.  When condition~\eqref{eqn:first-chain-condition} speaks of a ``set of curves'' passing through chains $I_1, \dots, I_c$, it is meant that for each  chain we have have a single curve passing through the beads of that chain, and that the curves do not intersect.  With this definition in place, the reader may now like to see the proof of Lemma~\ref{lem:catan-init} at the end of this section.

    Given such a set of curves, we can and will always perturb them slightly so that the only black or white diagram points that the curves pass through are points with beads.  As the curves do not intersect, we may order them as $Q_1, Q_2, \dots, Q_{c}$ from the southeasternmost to the northwesternmost; we can assume that the chains are renumbered correspondingly.  The curves thereby divide the plane into \emph{regions}, between the successive curves.

    Before entering properly into the main part of the proof of Theorem~\ref{thm:catan-conj} we need one more bit of terminology.  Whenever a curve intersects a jump line, we call the intersection point a ``crossing''.  We will further categorize each crossing as either ``horizontal'' or ``vertical'', according to whether the point is on a horizontal or vertical jump line segment.  (Cf.~the ``chain curve'' crossings in~\cite[Section~3.1]{Wer94}.)  In case the crossing is at a jump line corner (as happens when the curves passes through a bead), we classify the crossing as \emph{vertical}.

    We now come to the main part of the proof: In the geometric diagram we have chains of beads $I_1, \dots, I_c$ of total length~$L$, as well as a set of curves $Q_1, Q_2, \dots, Q_{c'}$  passing through them, with all inadmissible black points to the southeast of~$Q_1$.  Our proof will be algorithmic, with four \emph{phases}.
    \begin{itemize}
        \item \textbf{Phase 1:} In which new beads are added, and new white points are deemed admissible.
        \item \textbf{Phase 2:} In which beads are partly slid, to ``promote'' them to new chains.
        \item \textbf{Phase 3:} In which the new chains, $J_1, \dots, J_c$, are further slid to white points.
        \item \textbf{Phase 4:} In which a new set of curves is drawn through $J_1, \dots, J_c$, to satisfy~\eqref{eqn:first-chain-condition}.
    \end{itemize}
    We now describe each of the phases in turn.
    \paragraph{Phase 1.}  In this phase we consider the horizontal crossings, if any, of the first (i.e., southeasternmost) curve~$Q_1$.  For each horizontal crossing, we consider the white point immediately westward.  If that point is currently inadmissible, we declare it to be admissible, and we add a new bead at the crossing point.  Certainly this procedure introduces as many beads~$\Delta$ as it does newly admissible white points.  Also, although the new beads are not (yet) at points in the Viennot diagram, we \emph{can} add them to the first chain~$I_1$, in the sense that they fit into the increasing sequence of beads already in~$I_1$.  (This is because the curve~$Q_1$ is increasing from southwest to northeast.)      We now want to make the following claim:\\

    \emph{Key Claim: All white points to the northwest of $Q_1$ are now admissible}. \\

    (Recall that no white point is exactly on~$Q_1$.) To see the claim, consider any white point~$p$ to the northwest of $Q_1$.  Consider the jump line segment extending east from~$p$.  If that segment crosses~$Q_1$ then it's a horizontal crossing, and Phase~$1$ makes~$p$ admissible.  Otherwise, the segment must terminate at a black point~$q$ that is on or northwest of~$Q_1$.  (The segment cannot be a half-infinite extending eastward to infinity, because of the condition that~$Q_1$ eventually extends infinitely northward.)  By the precondition~\eqref{eqn:first-chain-condition}, $q$ must be an admissible black point.  Thus~$p$ is an admissible white point, being at the same height as~$q$.

    \paragraph{Phase 2.} For this phase it will be notationally convenient to temporarily add a ``sentinel curve'' $Q_{c+1}$ that is far to the northwest of the last curve~$Q_c$; the only purpose of this curve is to create vertical crossing points on all of the northward half-infinite segments of the jump lines.  We now proceed through each bead~$x$ in the diagram and potentially ``promote'' it to a higher-numbered chain.  The algorithm is as follows:  We imagine traveling west and north from~$x$ along its jump line until the first time that a vertical crossing point~$p$ is encountered.  (Note that such a vertical crossing point must always exist because of the sentinel curve~$Q_{c+1}$.)  Let~$q$ denote the crossing point on this jump line immediately \emph{preceding}~$p$.  (This~$q$ will either be the current location of~$x$, or it will be northwest of~$x$.)  We now slide bead~$x$ to point~$q$.  If~$x$ indeed moves (i.e., $q$ is not already its current location), we say that~$x$ has been ``promoted'' to a higher curve/chain.  As mentioned, we perform this operation for every bead~$x$.

    A crucial aspect of this phase is that the beads on a single jump line never ``pass'' each other, and in particular we never try to place two beads on the same crossing point. This is because whenever we have two consecutive beads on a jump line, there is always a vertical crossing at the higher bead.  (Of course, prior to Phase~1 all beads were at vertical crossings, by definition.  After Phase~1 we may have some beads at horizontal crossings, but at most one per jump line, and only on the lowest curve~$Q_1$.)

    At the end of Phase~2, all beads end up at (distinct) crossing points; in particular, they are all on the curves $Q_1, \dots, Q_c$.  (A bead cannot move onto the sentinel curve~$Q_{c+1}$.) Thus the beads may naturally partitioned into at most~$c$ chains (increasing sequences), call them $J_1, \dots, J_c$, some of which may be empty.  We now have:\\

    \emph{Phase 2 Postcondition: For every bead, the first crossing point on its jump line to its northwest is a vertical crossing.}

    \paragraph{Phase 3.}  The goal of Phase~$3$ is to further slide the (nonempty) chains $J_1, \dots, J_c$ northwestward along their jump lines so that they end up at white points, forming white chains. Since we will continue to move beads only northwestward, all the final white resting points will be admissible, by the \emph{Key Claim} above.  Another property of Phase~3 will be that each (nonempty) chain~$J_i$ will stay strictly inside the region between curves~$Q_i$ and $Q_{i+1}$.  Because of this, we will again have the property that beads will never slide past each other or end at the same white point, and the order in which we process the chains does not matter.

    So let us fix some (nonempty) chain~$J_i$ on curve~$Q_i$ and see how its beads can be slid northwestward along their jump lines to white points forming a chain southeast of~$Q_{i+1}$.  We begin with the northeasternmost bead on the chain, which (by virtue of Phase~$2$) is either on a black point or is in the middle of a horizontal jump line segment.  In either case, we begin by sliding it west, and we deposit immediately at the first white point encountered.  We must check that this white point is still to the southeast of curve $Q_{i+1}$. This is true because otherwise the first crossing point northwest of the bead's original position would be a horizontal one, in contradiction to the \emph{Phase~2 Postcondition.}

    We handle the remaining beads in $J_i$ inductively.  Suppose we have successfully deposited the northeasternmost~$t$ beads of $J_i$ at white points that form a chain southeast of $Q_{i+1}$.  Let's say the last of these (the southwesternmost of them) is bead~$x_t$, and we now want to successfully slide the next one, $x_{t+1}$.  Let $p_t$ denote the white point into which~$x_t$ was slid.  Our method will be to slide $x_{t+1}$ west and north along its jump line until the first time it encounters a white point, $p_{t+1}$ that is west of~$p_t$, depositing it there. We need to argue three things: (i)~$p_{t+1}$ exists; (ii)~$p_{t+1}$ is southeast of the curve $Q_{i+1}$; (iii)~$p_{t+1}$ is south of~$p_t$.  If these things are true then we will have deposited $x_{t+1}$ in such a way that it extends the white point chain and is still southeast of $Q_{i+1}$. This will complete the induction.

    To check the properties of $p_{t+1}$, consider also the last (northwesternmost) white point $p'_{t+1}$ on $x_{t+1}$'s jump line that is still southeast of $Q_{i+1}$. At least one must exist because the jump line crosses $Q_{i+1}$ vertically, by the \emph{Phase~2 Postcondition}.  This $p'_{t+1}$ must be west of $p_t$, as otherwise the jump line segment extending north from from it would cross $p_t$'s jump line.  It follows that $p_{t+1}$ must exist and be southeast of $Q_{i+1}$.  It remains to check that $p_{t+1}$ is indeed south of $p_t$. But if this were not true then bead $x_{t+1}$ would have slid north of $p_t$ prior to sliding west of it --- impossible, as again it would imply $x_{t+1}$'s jump line crossing $x_{t}$'s.  This completes the induction, and the analysis of Phase~3.

    \paragraph{Phase 4.}  Here we need to show that we can draw a new set of curves through the final (nonempty) chains $J_1, \dots, J_c$ such that condition~\eqref{eqn:first-chain-condition} is satisfied.  This is rather straightforward.  As shown in Phase~3, each final chain~$J_i$ is confined to a region between the old curves $Q_{i}$ and $Q_{i+1}$.  Thus there is no difficulty in drawing new nonintersecting curves through these chains, also confined within the regions, that pass through all the beads. Finally, as the ``new first curve'' is completely to the northwest of the ``old first curve''~$Q_1$, the fact that all inadmissible white points are southeast of it follows from the \emph{Key Claim}.
\end{proof}

\begin{proof}[Proof of Lemma~\ref{lem:catan-init}]
    Suppose we are given initial chains of beads $I_1, \dots, I_c$.  As the beads on each chain are increasing, southwest-to-northeast, there is no difficulty in drawing a curve through each chain.  The only catch, for the purposes of getting a ``set of curves'', is that the curves might intersect.  (We can assume by perturbation, though, that two curves never intersect at a bead.)  To correct the intersections, consider any two curves $Q_i$ and $Q_j$ that intersect.  Redefine these curves as follows: take all the ``lower'' (southeastern) segments and call that one new curve, and take all the ``upper'' (northwestern) segments and call that another new curve.  All of the beads are still on curves, and thus can still be repartitioned into~$c$ chains.  The resulting new curves still technically have points of contact but do not essentially cross each other; by a perturbation we can slightly pull apart the points of contact to make sure they become truly nonintersecting.

    (An alternative to this proof is the following:  For our overall proof of Theorem~\ref{thm:John-conjecture} we may as well assume that the initial set of~$c$ disjoint sequences in~$b$ is of maximum possible total length.  In this case, Wernisch's maximum $c$-chain algorithm~\cite[Theorem~21]{Wer94} provides a set of increasing, nonintersecting ``chain curves'', together with a maximizing set of $c$~chains that are in the associated ``regions''. This makes it easy to draw a nonintersecting set of curves through the chains, as in Phase~4 above.)
\end{proof}

\bibliographystyle{alpha}
\bibliography{tomography}

\end{document}